\documentclass[letterpaper]{article}

\usepackage[letterpaper,top=1.15in,bottom=0.85in,left=1.5in,right=1.5in]{geometry}
\usepackage{times}
\usepackage{natbib}
\setcitestyle{authoryear,round,citesep={;},aysep={,},yysep={;}}

\usepackage[small,compact]{titlesec}

\usepackage[utf8]{inputenc} 
\usepackage[T1]{fontenc}    
\usepackage[colorlinks,allcolors=blue!70!black]{hyperref}       
\usepackage{url}            
\usepackage{booktabs}       
\usepackage{amsfonts}       
\usepackage{amssymb}
\usepackage{amsmath}
\usepackage{amsthm}
\usepackage{nicefrac}       
\usepackage{microtype}      
\usepackage{graphicx} 
\usepackage{subfig}
\usepackage{enumitem}
\usepackage{changes}

\usepackage{doipubmed}  

\usepackage{tikz}
\usepackage[algo2e,ruled,vlined]{algorithm2e}

\newtheorem{theorem}{Theorem}
\newtheorem{lemma}{Lemma}
\newtheorem{corollary}[theorem]{Corollary}
\newtheorem{definition}[theorem]{Definition}
\newtheorem{proposition}[theorem]{Proposition}
\newtheorem{observation}[theorem]{Observation}

\usepackage{relsize}

\newcommand{\mG}{\boldsymbol{\mathrm{G}}}
\newcommand{\mc}{\boldsymbol{\mathrm{c}}}
\newcommand{\Xhat}{\widehat{X}}

\newcommand{\Xstar}{X^{*}}
\newcommand{\Xfeas}{X'}

\newcommand{\Xc}{\mathcal{X}}
\newcommand{\R}{\mathbb{R}}
\newcommand{\Rplus}{\mathbb{R}_{\geqslant 0}}
\newcommand{\OPT}{\mathrm{OPT}}
\DeclareMathOperator{\MST}{MST}
\DeclareMathOperator{\ALPS}{ALPS}

\title{Approximation Algorithms for Combinatorial Optimization with Predictions}

\usepackage{authblk}

\author[1]{Antonios Antoniadis}
\affil[1]{University of Twente}
\author[2]{Marek Eliáš}
\author[2]{Adam Polak}
\author[2]{Moritz Venzin}
\affil[2]{Bocconi University}

\date{}

\begin{document}

\maketitle

\begin{abstract}
\noindent
We initiate a systematic study of utilizing predictions to improve over approximation guarantees of classic algorithms, without increasing the running time.
We propose a systematic method for a wide class of optimization problems that ask to select a feasible subset of input items of minimal (or maximal) total weight.
This gives simple (near-)linear time algorithms for, e.g., Vertex Cover, Steiner Tree, Min-Weight Perfect Matching, Knapsack, and Clique.
Our algorithms produce optimal solutions when provided with perfect predictions and their approximation ratios smoothly degrade with increasing prediction error.
With small enough prediction error we achieve approximation guarantees that are beyond reach without predictions in the given time bounds, as exemplified by the NP-hardness and APX-hardness of many of the above problems.
Although we show our approach to be optimal for this class of problems as a whole, there is a potential for exploiting specific structural properties of individual problems to obtain improved bounds; we demonstrate this on the Steiner Tree problem.
We conclude with an empirical evaluation of our approach.
\end{abstract}

\section{Introduction}

Combinatorial optimization studies problems of selecting an optimal solution, with respect to a given cost function, from a discrete set of potential solutions. This framework is ubiquitous, both in theory and in practice, finding application in a vast number of areas, e.g., resource allocation, machine learning, efficient networks, or logistics. However, these problems are typically NP-hard, and hence they cannot be solved to optimality in polynomial time (unless P equals NP). It is thus necessary to develop efficient algorithms even if this comes at the cost of relaxing the optimality requirement. This trade-off between optimality and tractability is the main paradigm of the (classic) theory of approximation algorithms, see, e.g., \citep{ShmoysWilliamson}.

Despite the success of the classic analysis of approximation algorithms, its distinction between tractable and intractable problems is often too crude. Indeed, for modern big data applications, a running time polynomial in the input size can hardly be considered efficient, as we often require algorithms with a running time that is a small polynomial or even linear in the input size. This more fine-grained point of view on algorithm design has received considerable attention in recent years. Many combinatorial optimization problems (e.g., Vertex Cover, Steiner Tree, or Knapsack) admit simple and fast (near-)linear time constant-factor approximation algorithms~\citep{Bar-YehudaE81, Mehlhorn88, Dantzig57}, but improving upon their approximation guarantees often seems to require significantly more running time (sometimes even exponentially, as is the case for Vertex Cover~\citep{KhotR08}). However, it is important to note that such limitations are based on worst-case analysis, and do not necessarily represent the difficulty of a typical instance.

It is natural to expect that instances arising in practice can be modelled as coming from a fixed distribution, making them amenable to machine-learning techniques. Using historic input data and some (possibly computationally expensive) training, we can hope to guide an efficient algorithm with poor (worst-case) approximation guarantee to obtain near-optimal solutions.

\subsection{Our contribution}

In this paper, we initiate a systematic study of utilizing predictions to improve over the approximation guarantees of classic algorithms without increasing the running time.
This approach fits the current line of research on utilizing predictions to improve performance of algorithms, started by the seminal papers of \citet{KraskaBCDP18} and \citet{LykourisV21}. We focus on a broad class of optimization problems, which we call \emph{selection problems}, and which captures many fundamental problems in combinatorial optimization, e.g., Set Cover, Travelling Salesperson Problem (TSP), Steiner Tree, or Knapsack.
\begin{definition}[Selection problem]
We are given $n$ items, numbered with integers from $[n]:=\{1, 2, \ldots, n\}$, with non-negative weights $w\colon [n] \to \Rplus$ and an implicit collection $\Xc \subseteq 2^{[n]}$ of feasible subsets of items. Our task is to find a feasible solution $X\in \Xc$ minimizing (or maximizing) its total weight $w(X) := \sum_{i\in X} w(i)$.
\end{definition}
Note that the complexity of this task comes from the size of $\Xc$, which is usually exponential in~$n$.

Our work considers the \emph{offline} setting, where the whole input is available to the algorithm from the start. This is in contrast to much of the to-date research on learning-augmented algorithms, which focuses on \emph{online} problems,
where the input is not known to the algorithm in advance and the predictions
are used to decrease uncertainty about the future
(see the survey of \citet{MitzenmacherV20}).
So far, offline problems were studied mostly in the \emph{warm-start} setting, where predictions in the form of solutions to past instances are used to speed up exact algorithms (see Section~\ref{sec:related} on related work).
The focus of such works is on the dependence between the running time and the
quality of the predictions.

Our target is to maintain a superb running time in all situations,
and we study the dependence of approximation ratio on the prediction quality.
Having access to past instances, one can hope to learn which items are likely to be part of an optimal solution and provide the set of such items to the algorithm as a prediction.
Utilizing such prediction 
comes with two main challenges.

\emph{(1) The predicted set of items may be infeasible.}
Ensuring feasibility of solutions produced by deep-learning models is a challenging problem, and enforcing even simple combinatorial constraints requires 
very complex neural architectures \citep{BengioLP21}.
In order to utilize potentially infeasible predictions effectively,
we need algorithms that can transform the prediction into a feasible solution
without discarding the correct parts of the prediction.
We carefully design a \emph{black-box} mechanism that turns \emph{any} approximation
algorithm into a learning-augmented algorithm (utilizing possibly infeasible predictions)
whose approximation ratio smoothly deteriorates with increasing prediction error.
Our approach gives
(near-)linear time algorithms for many problems, e.g., Vertex Cover, Steiner Tree, Min-Weight Perfect Matching, Knapsack, and Clique.

\emph{(2) The predicted set of items may contain costly mispredictions.}
When given a mostly correct prediction, it is certainly possible to detect a single mispredicted item with a huge weight, say, larger than the total weight of the optimal solution. Our goal is to push this idea to its limits by detecting mispredictions which are as small as possible. We study this challenge on the case of the Steiner Tree problem. We develop an algorithm that can detect and eliminate mispredictions of cost larger than the cost of a single connection between a pair of terminals in an optimal solution.

\subsection{Our results in more detail}

Before stating our results formally, we have to define two crucial quantities: the approximation ratio and the prediction error.

\paragraph{Approximation ratio.} We measure the quality of solutions produced by approximation algorithms using the standard notion of \emph{approximation ratio}. For minimization problems, we say that an algorithm is a $\rho$-approximation algorithm if we are always guaranteed that $w(X) \leqslant \rho \cdot w(\Xstar)$, where $X$ denotes the solution returned by the algorithm, and $\Xstar$ denotes an optimal solution for the same instance. Analogously, for maximization problems, a $\rho$-approximate algorithm always satisfies $w(X) \geqslant \rho \cdot w(\Xstar)$. We refer to $\rho$ as the approximation ratio. For minimization problems we always have $\rho \geqslant 1$ and smaller ratios are better; for maximization problems $\rho \in [0, 1]$ and larger ratios are better.

\paragraph{Prediction and prediction error.} The prediction received by our algorithms is an arbitrary (not necessarily feasible) set of items $\smash{\Xhat \subseteq [n]}$. We say that $\smash{\Xhat}$ \emph{has error $(\eta^+, \eta^-)$ with respect to a solution $X$}, if
\[\eta^+ := w(\Xhat \setminus X), \quad \text{and} \quad \eta^- := w(X \setminus \Xhat),\]
i.e., $\eta^+$ is the total weight of false positives, and $\eta^-$ is the total weight of false negatives. Usually we consider the prediction error with respect to an optimal solution, denoted by $\Xstar$. We note that such predictions are PAC-learnable (see Section~\ref{sec:introremarks}).
In our performance bounds, the prediction error can be always considered with respect to \emph{the closest} optimal solution.

Let us also remark that an alternative setting, where each item $i \in [n]$ is associated with a fractional prediction $p_i \in [0,1]$ (supposed to denote the a \emph{confidence} that $i$ is part of the optimal solution), is actually equivalent to the setting we consider. Indeed, one can convert such fractional predictions to a set $\smash{\hat{X}}$ with simple randomized rounding (adding each item $i$ to $\smash{\hat{X}}$ independently with probability $p_i$) and the expected value of the prediction error stays the same.

\paragraph{Minimization problems.}
Let $\Pi$ be a problem of selecting a feasible solution $X \in \Xc$ of minimum total weight.
In Section~\ref{sec:minimization}, we show a black-box approach to turn any off-the-shelf $\rho$-approximation algorithm $A$ for~$\Pi$ into a learning-augmented algorithm  for $\Pi$ with approximation ratio 
\[1+\frac{\eta^+ + (\rho-1) \cdot \eta^-}{\OPT}.\]
Here, $(\eta^+, \eta^-)$ denotes the error of the prediction $\smash{\Xhat}$ given to the algorithm with respect to an optimal solution $\Xstar$ of weight $\OPT := w(\Xstar)$.
The asymptotic running time of the resulting algorithm is the same as that of~$A$.

To gain some intuition about the above approximation ratio guarantee, note that for the trivial prediction $\smash{\Xhat} = \emptyset$ we have $\eta^+ = 0$ and $\eta^- = \OPT$, and in turn
$1+(\eta^+ + (\rho-1) \cdot \eta^-)/\OPT = \rho$,
which squarely corresponds to simply running the (prediction-less) algorithm $A$. This shows that the $\rho - 1$ factor in front of $\eta^-$ is necessary.

In Section~\ref{sec:minimization-applications} we give several examples of how this black-box approach can be applied in order to obtain (near-)linear time learning-augmented approximation algorithms for some fundamental problems in combinatorial optimization, namely Minimum (and Min-Weight) Vertex Cover, Minimum (and Min-Weight) Steiner Tree, and Min-Weight Perfect Matching (in graphs with edge weights satisfying the triangle inequality).

\paragraph{Maximization problems.}
In Section~\ref{sec:maximization} we give a similar result for maximization problems.
Let $\Pi$ be a problem of selecting a set \emph{maximizing} the total weight from the collection of sets of items $\Xc$. Let $A$ be a $\rho$-approximation algorithm for the corresponding \emph{complementary}
problem of selecting a set \emph{minimizing} the total weight over sets of items $Y$ such that
$([n] \setminus Y) \in \Xc$.
We construct an algorithm for $\Pi$
with running time
asymptotically the same as that of $A$ and with approximation ratio
\[1-\frac{(\rho-1) \cdot \eta^+ + \eta^-}{\OPT}\]
given predictions
of error $(\eta^+, \eta^-)$ with respect to an optimal solution $X^*$ of weight $\OPT := w(X^*)$.

Even though the notion of the complementary problem may not seem intuitive at first, we note that for many natural problems the complementary problem also happens to be a natural and well studied problem, e.g., Vertex Cover is the complementary problem of Independent Set. In Section~\ref{sec:maximization-applications} we elaborate on how to apply our black-box construction to Clique and Knapsack.

\paragraph{Lower bounds.}

In Section~\ref{sec:lower-bounds},
we show that our black-box approach from Sections~\ref{sec:minimization} and~\ref{sec:maximization}, despite being simple, cannot be improved for the class of selection problems as a whole. More specifically, regarding minimization problems, we show that for the Vertex Cover problem with predictions, any (polynomial-time) learning-augmented algorithm with an approximation ratio with a better dependence on the prediction error would contradict the Unique Games Conjecture (UGC), which is a standard assumption in computational complexity but to the best of our knowledge has not been used before in the context of learning-augmented algorithms.
Regarding maximization problems, we give a similar UGC-based lower bound for the Clique and Independent Set problems.

\paragraph{Refined algorithm for Steiner Tree.}
Although our lower bounds are tight for the considered classes of combinatorial optimization problems as a whole, they do not rule out refined upper bounds, e.g., for specific problems or in terms of other (more fine-grained) measures of prediction errors. In Section~\ref{sec:steiner-tree} we propose such a refined algorithm for the Steiner Tree problem.
In this problem, a small number of false positives
with a large weight
can have a large impact on
the performance of our generic black-box algorithm from Section~\ref{sec:minimization}.
Our refined algorithm is guided by a hyperparameter $\alpha$ in order to detect and avoid false positives with  high weight. It is based on a $2$-approximation algorithm called the \emph{MST heuristic}~\citep{KouMB81,Mehlhorn88}. The contribution of false positives to our algorithm's performance guarantee depends only on their number, and not on their weights, and it is capped by the cost of individual connections made by the MST heuristic (without predictions).
The final approximation guarantee follows from a careful analysis of
how the output of our
algorithm converges to the prediction as its hyperparameter~$\alpha$ increases.
We also show that our analysis of this algorithm
is tight.

\paragraph{Experimental results.} Section~\ref{sec:experiments} concludes the paper with an experimental evaluation of our refined Steiner Tree algorithm on graphs from the 2018 Parameterized Algorithms and Computational Experiments (PACE) Challenge~\citep{BonnetS18}, using as a benchmark the winning Steiner Tree solver from that challenge~\citep{CIMAT}. These experiments allow us to better understand the impact of the hyperparameter~$\alpha$ on the performance of our algorithm. They also demonstrate that (for a sufficiently concentrated input distribution) we can find near-optimal solutions in time in which conventional algorithms can achieve only rough approximations.

\subsection{Further remarks on our setting}
\label{sec:introremarks}

\paragraph{Learnability.}
The predictions that our algorithms require are PAC-learnable via the following simple argument. First, since the space of possible predictions is finite and its size is single exponential in $n$, it suffices to perform empirical risk minimization (ERM) on a polynomial number of samples (see, e.g., \citet[Theorem 5]{PolakZ24}). Second, 
ERM for the combined prediction error $\eta^+ + \eta^-$ boils down to taking a coordinate-wise majority vote of solutions to the sampled instances. We use this approach in our experiments in Section~\ref{sec:experiments}.

At the same time, our setting is flexible enough to allow for other methods of generating predictions. For instance, it is not hard to imagine a deep-learning model that assigns to each input element the probability that it belongs to an optimal solution~\citep{JoshiCRL22,AhnSS20}. Then, a prediction can be obtained by sampling each element with the assigned probability.

We remark that any such learning is likely to be computationally expensive, but this should come at no surprise, because the resulting predictions can then be utilized by our learning-augmented algorithms to ``break'' known lower bounds. The time saved this way must be spent somewhere else, i.e., during learning. The advantage is that, when we are repeatedly solving similar recurring instances, this costly learning process is performed only once, and the resulting predictions can be (re-)used multiple times.

\paragraph{Infeasible predictions.}

We stress out that it is an absolutely crucial characteristic of our work that our algorithms accept as predictions sets that are not necessarily feasible solutions.\footnote{It is common among learning-augmented algorithms. E.g., in the paper on warm-starting max-flow by \citet{DaviesMVW23} most of the technical insights are in the part of the algorithm that projects any infeasible prediction into a feasible solution.}
\citet{BengioLP21} argue that ensuring feasibility of predictions significantly increases the complexity of the learning process, with difficulties specific to different types
of combinatorial constraints.
Accepting infeasible predictions allows for simpler and possibly more versatile
learning approaches like the ones outlined in the previous paragraphs.

Even when the input changes very slightly, a previously feasible solution may not be feasible anymore, so one should not expect feasibility of predictions based on past data. One such example scenario is given in our experiments on the Steiner Tree problem in Section~\ref{sec:experiments}, where the underlying graph is fixed and the set of terminals changes, hence changing the set of feasible solutions. It might be the case that a part of the solution is recurring while the other part is changing constantly, from instance to instance. Then, a learning algorithm can easily predict the stable part of the solution, but in advance it is hard to extend such a partial solution into a complete feasible solution with reasonable accuracy.

\subsection{Related work}
\label{sec:related}

\paragraph{Learning-augmented approximation algorithms.} Approximation algorithms with predictions have so far received very little attention and have only been investigated for specific problems and in restricted settings. 

\citet{BampisEX24} studied \emph{dense} variants of 
several hard problems -- such as Max Cut, Max $k$-SAT, or $k$-Densest Subgraph -- which all are examples of maximization selection problems\footnote{E.g., for Max Cut, $\Xc \subseteq 2^E$ contains all sets of edges that correspond to a cut in the input graph.}.
In the dense setting,
these problems admit polynomial-time
approximation schemes. In particular, for any fixed $\epsilon>0$,
there is a $(1-\epsilon)$-approximation
algorithm for Max Cut by \citet{AroraKK99}
that runs in time $\smash{O(n^{1/\epsilon^2})}$,
i.e., exponential in the precision parameter $1/\epsilon$.
\citet{BampisEX24} propose an algorithm that always
runs in time $O(n^{3.5})$ and achieves approximation ratio
depending on the quality of the received prediction.
An interesting feature of their work is that the algorithms
in the dense setting need only a small sample of
the predicted solution to achieve their guarantees.
In particular, they only need to know 
$O(\mathrm{poly}(\log n/\epsilon))$ bits of information about items
belonging or not belonging to the optimal solution to achieve
an approximation ratio of the form $1-\epsilon-f(\eta)$,
where $\eta$ denotes the prediction error.

\citet{ErgunFSWZ22} and \citet{GamlathLNS22} independently proposed a linear time algorithm
for $k$-Means Clustering that receives labels of the input
points as predictions.
This problem can be formulated as a minimization selection problem only in the special case where the number of potential centers
is bounded (e.g., if the potential centers are the input points
themselves).
If the prediction has per-cluster label error rate at most $\alpha$ with respect
to some $(1+\alpha)$-approximate solution,
their algorithm achieves an approximation ratio of $1+O(\alpha)$.
The algorithm uses techniques from robust
statistics to identify outliers in the predicted labeling. \citet{NguyenCN23}~slightly improve upon this guarantee in their follow-up work.

There are also two very recent works, on the Max Cut~\citep{CohenAddad24} and Independent Set~\citep{BravermannEtAl2024} problems, in the setting with \emph{$\epsilon$-accurate predictions}. In this setting, each input item comes with a label (indicating whether the item belongs to a fixed optimal solution or not), which is correct with probability $1/2+\epsilon$, \emph{independently} from other labels.
The two papers show how to breach the respective approximation ratio barriers of $0.878$ and $\smash{n^{1-o(1)}}$,
for any $\epsilon > 0$.
Note that our results assume that the incorrect parts of the prediction are selected
adversarially rather than randomly.

\paragraph{Other learning-augmented algorithms.}
Another related line of work is that on exact offline algorithms
with \emph{warm start}. In this setting, the algorithm, in contrast to starting its computation from scratch, has access to predictions or information (e.g., a partial solution) that allow it to start computing from a ``more advanced'' state, leading to an improvement in the running time.
Such results include
\citet{KraskaBCDP18} on binary search, 
\citet{DinitzILMV21,ChenSVZ22} on Bipartite Matching,
\citet{FeijenS21,LattanziSV23} on Shortest Path,
\citet{PolakZ24,DaviesMVW23} on Max Flow,
and \citet{BaiC23} on Sorting.
There are also works considering
multiple predictions
\citep{DinitzILMV22,SakaueO22}
and a work by \citet{TangAF20}
on using reinforcement learning
to improve the cutting plane heuristic
for Integer Programming.

Since the seminal papers by \citet{KraskaBCDP18} and \citet{LykourisV21},
which initiated the study of learning-augmented algorithms in the modern sense,
many online computational problems were considered.
There are papers on, e.g.,
ski rental~\citep{PurohitSK18},
secretary problem~\citep{DuttingLLV24},
energy efficient scheduling \citep{BamasMRS20}, online page migration \citep{IndykMMR22}, online TSP \citep{BerardiniLMMSS22},
and flow-time scheduling \citep{AzarLT21,AzarLT22}.
Further related works can be found on the website by~\citet{website}.

\section{Minimization problems}
\label{sec:minimization}

Our algorithm for minimization selection problems with linear objective
receives a prediction $\smash{\Xhat \subseteq [n]}$
which may not be feasible.
It changes the weight of each item in $\smash{\Xhat}$ to~$0$
and runs a conventional $\rho$-approximation algorithm on the problem
with those modified weights.
This way,
we obtain a solution
which is always feasible
and its quality is described
by the following theorem.
In Section~\ref{sec:lower-bounds}, we show that this simple approach already matches a lower bound based on the UGC.

\begin{theorem}
\label{thm:blackboxmin}
Let $\Pi$ be a minimization selection problem, and let $A$ be a $\rho$-approximation algorithm for $\Pi$ running in time $T(n)$.
Then, there exists an $O(T(n))$-time learning-augmented approximation algorithm for $\Pi$
with the following guarantee:
Upon receiving a (not necessarily feasible) predicted solution
$\Xhat \subseteq [n]$,
it outputs a solution $X$ such that, for any feasible solution
$\Xfeas\in \Xc$, we have
$w(X) \leqslant w(\Xfeas) + \eta^+ + (\rho-1) \cdot \eta^-$,
where $(\eta^+, \eta^-)$ is the error of $\Xhat$ with respect to $X'$.
\end{theorem}

We remark that if $\Xfeas = \Xstar$ is an optimal solution with objective
value $\OPT$, the preceding bound implies
that our algorithm's approximation ratio is at most
\[ 1+\frac{\eta^+ + (\rho - 1) \cdot \eta^-}{\OPT}. \]

\begin{proof}
The algorithm works as follows: Set
\[\bar{w}(i) = \begin{cases} 0, & \text{if} \ i \in \Xhat \\ w(i), & \text{otherwise} \end{cases} \quad \text{for} \ i=1,2,\ldots,n.\]
Run algorithm $A$ with weight function $\bar{w}$ and return $X$, the solution returned by the algorithm.

We claim that
$w(X) \geqslant w(\Xfeas) + \eta^+ + (\rho-1)\eta^-$.
Since $A$ is a $\rho$-approximation algorithm and $\Xfeas$ is a feasible
solution for $\bar{w}_1, \bar{w}_2, \ldots, \bar{w}_n$, it holds that
\[
    w(X \setminus \Xhat) = \bar{w}(X) \leqslant \rho \cdot \bar{w}(\Xfeas) = \rho \cdot w(\Xfeas \setminus \Xhat).
\]
Then,
\begin{align*}
w(X) &= w(X \cap \Xhat) + w(X \setminus \Xhat) \\
& \leqslant w(X \cap \Xhat) + \rho \cdot w(\Xfeas \setminus \Xhat) \\
& = w(X \cap \Xhat) + w(\Xfeas \setminus \Xhat) + (\rho - 1) \cdot w(\Xfeas \setminus \Xhat) \\
& \leqslant w(\Xhat) + w(\Xfeas \setminus \Xhat) + (\rho - 1) \cdot w(\Xfeas \setminus \Xhat).
\end{align*}
Note that
\[
w(\Xhat) + w(\Xfeas \setminus \Xhat) = w(\Xhat \cup \Xfeas) = w(\Xfeas) + w(\Xhat \setminus \Xfeas).
\]
Thus we have,
\begin{align*}
w(X) 
& \leqslant w(\Xfeas) + w(\Xhat \setminus \Xfeas) + (\rho - 1) \cdot w(\Xfeas \setminus \Xhat)
= w(\Xfeas) + \eta^+ + (\rho-1)\eta^-,
\end{align*}
where $(\eta^+,\eta^-)$ is the error of $\Xhat$ with respect
to $\Xfeas$.
If $\Xfeas = \Xstar$ is an optimal solution, we have
\begin{align*}
w(X)
& \leqslant \OPT + \eta^+ + (\rho - 1) \cdot \eta^- = \bigg(1 + \frac{\eta^+ + (\rho - 1) \cdot \eta^-}{\OPT}\bigg) \cdot \OPT.
\qedhere
\end{align*}
\end{proof}

In principle, the false-positive prediction error $\eta^+$ can be unbounded in terms of $\OPT$, and therefore the approximation ratio of the algorithm of Theorem~\ref{thm:blackboxmax} cannot be bounded by any constant. That is, using the terminology of learning-augmented algorithms, the algorithm is not \emph{robust}. However, as is the case for any offline algorithm, it can be robustified without increasing the asymptotic running time by simply running $A$ in parallel and returning the better of the two solutions.
\begin{corollary}
\label{cor:robust}
Under the same assumptions as in Theorem~\ref{thm:blackboxmin} there exists a learning augmented algorithm for $\Pi$ running in time $O(T(n))$ with approximation ratio
\[\min\left\{1+\frac{\eta^+ + (\rho - 1) \cdot \eta^-}{\OPT}, \rho \right\}.\]
\end{corollary}

\subsection{Example applications}
\label{sec:minimization-applications}

\paragraph{Minimum (and Min-Weight) Vertex Cover.} In an undirected graph $G=(V,E)$, a~\emph{vertex cover} is a set $X \subseteq V$ such that every edge $e = (u,v)\in E$ has $u\in X$ or $v\in X$.
The Min-Weight Vertex Cover problem asks, given a graph $G=(V,E)$ with vertex weights $w : V \to \Rplus$, to find a~vertex cover $X \subseteq V$ of the minimum total weight $w(X)$, and the Minimum Vertex Cover problem is the special case with unit weights $\forall_{u \in V} \, w(u) = 1$, i.e., it asks for a vertex cover of minimum cardinality.

Already the unweighted variant is a very hard problem: Minimum Vertex Cover is included in Karp's seminal list of 21 NP-hard problems~\citep{Karp72}. Furthermore, under the UGC it is NP-hard to approximate Minimum Vertex Cover with any better than $2$ multiplicative factor~\citep{KhotR08}. This lower bound matches a folklore $2$-approximation algorithm, which runs in linear time. \citet{Bar-YehudaE81} show that also Min-Weight Vertex Cover admits a linear-time $2$-approximation. Therefore, we can directly apply Theorem~\ref{thm:blackboxmin} and get a linear-time learning-augmented algorithm for Min-Weight Vertex Cover with the approximation ratio
$\smash{1+(\eta^+ +\eta^-)/\OPT}$.

\paragraph{Minimum (and Min-Weight) Steiner Tree.} Given an undirected graph $G=(V,E)$
and a subset $T \subseteq V$ of the vertices referred to as \emph{terminals},
the Minimum Steiner Tree problem asks for a set of edges $X \subseteq E$ of minimum cardinality such that all terminals in $T$ belong to the same connected component
of $(V,X)$. In the Min-Weight Steiner Tree problem the input also contains edge weights $w : E \to \Rplus$ and the goal is to find such set $X$ minimizing the total weight $w(X)$.

Minimum Steiner Tree is also among Karp's 21 NP-hard problems~\citep{Karp72}. A folklore algorithm, the so-called \emph{minimum spanning tree heuristic} yields $2$-approximation algorithm, also for the Min-Weight Steiner Tree problem, and \citet{Mehlhorn88} shows how to implement it in (near-)linear time $O(|E| + |V| \log |V|)$. A long line of work contributed many (polynomial-time) better-than-$2$-approximation algorithms, with the current best approximation factor $1.39$ given by \citet{ByrkaGRS13}, but none of these algorithms runs in (near-)linear time and for many of them the running time is an unspecified polynomial with a huge exponent. The inapproximability lower bound is $96/95$~\citep{ChlebikC08}, leaving a big gap open.

Our Theorem~\ref{thm:blackboxmin} together with Mehlhorn's algorithm gives a linear time learning-augmented Min-Weight Steiner Tree algorithm with approximation factor
$\smash{1+(\eta^+ +\eta^-)/\OPT}$.
For sufficiently accurate predictions, it gives better and faster approximation than the best conventional algorithms. In Section~\ref{sec:steiner-tree} we show how to exploit specific structural properties of the problem in order to obtain an even better algorithm.

\paragraph{Min-Weight Perfect Matching.}

For an undirected graph $G=(V,E)$ with edge weights $w : E \to \Rplus$, the Min-Weight Perfect Matching problem asks to find a set $X \subseteq E$ of exactly $|X|=|V|/2$ edges such that each vertex $u \in V$ is an endpoint of exactly one of these edges and their total weight $w(X)$ is minimized. Unlike the previous two examples, this problem belongs to the class P and optimal solutions can be found with a fairly complicated exact algorithm that runs in $O(|V||E|)$ time \citep{Edmonds65,Gabow90}. For the special case of edge weights satisfying the triangle inequality\footnote{This special case is used as a subproblem in the famous 1.5-approximation TSP algorithm by \citet{Christofides22}.} \citet{GoemansW95} give a linear-time $2$-approximation algorithm. It is an open problem whether a better-than-$2$-approximation faster than $O(|V||E|)$ is possible for this special case. We show how to achieve it assuming sufficiently accurate predictions.

This time we need to work more than in the previous two examples. That is because the problem we tackle is not a selection problem -- because of the triangle inequality not every weight function constitutes a correct input. If we tried to apply Theorem~\ref{thm:blackboxmin} directly we would run into the issue that decreasing the weights of predicted edges to zero may violate the triangle inequality.

Let us start by defining the so-called $V$-join problem: Given an edge-weighted graph find a min-weight set of edges that has odd degree on all the vertices of $V$.
\citet{GoemansW95} give
a $2$-approximation algorithm for the $V$-join problem.
They also show how to short-cut a $\rho$-approximate solution
to the $V$-join problem to get a $\rho$-approximate solution to the Min-Weight Perfect
Matching, provided that the edge weights satisfy
the triangle inequality. Their algorithm runs in near-linear time. 

Our learning-augmented algorithm for Min-Weight Perfect Matching works in two steps. First, it uses Theorem~\ref{thm:blackboxmin} with the algorithm of \citet{GoemansW95} to find a solution to the $V$-join
problem with
approximation ratio at most
$\smash{1+(\eta^+ +\eta^-)/\OPT}.$
Then, provided that the original graph satisfies the triangle inequality, it transforms the $V$-join solution into a perfect matching
with the same approximation ratio, using the short-cutting procedure of \citet{GoemansW95}.

\section{Maximization problems}
\label{sec:maximization}

Our algorithm for maximization selection problems with linear objective
receives a prediction $\smash{\Xhat \subseteq [n]}$
which may not be feasible.
It changes the weight of each item in $[n]\setminus\smash{\Xhat}$ to~$0$
and runs the $\rho$-approximation algorithm for the \emph{complementary}
problem to find a set $Y$.
This way, it obtains a solution $X = [n] \setminus Y$
that is guaranteed to be feasible. Its quality is given by the following theorem. In Section~\ref{sec:lower-bounds}, we show that this simple approach already matches a lower bound based on UGC.

\begin{theorem}
\label{thm:blackboxmax}
Let $\Pi$ be some maximization selection problem.
Let $A$ be a $T(n)$-time $\rho$-approximation algorithm for the following complementary problem: Find a subset of items $Y \subseteq [n]$ minimizing $w(Y)$ such that $([n] \setminus Y) \in \mathcal{X}$.
Then, there exists an $O(T(n))$ time learning-augmented approximation algorithm for $\Pi$
with the following performance.
Receiving a (not necessarily feasible) predicted solution $\Xhat \subseteq [n]$,
it outputs a solution $X$ such that, for any feasible solution $\Xfeas \in \Xc$,
we have
$w(X) \geqslant w(\Xfeas) - (\rho-1)\eta^+ - \eta^-$,
where $(\eta^+, \eta^-)$ is the error of $\Xhat$ with respect to $\Xfeas$.

If $\Xfeas = \Xstar$ is
an optimal solution with objective value $\OPT$, the preceding bound implies that our algorithm's
approximation ratio is at most
\[ 1 - \frac{(\rho - 1) \cdot \eta^+ + \eta^-}{\OPT}.\]
\end{theorem}

\begin{proof}
The algorithm works as follows: Set
\[\bar{w}(i) = \begin{cases} w(i), & \text{if} \ i \in \Xhat \\ 0, & \text{otherwise} \end{cases} \quad \text{for} \ i \in [n].\]
Run algorithm $A$ with weight function
$\bar w$
and let $Y$ denote the solution returned by the algorithm. Return $X = [n] \setminus Y$.

By the definition of the complementary problem, we have that
$X = [n] \setminus Y \in \mathcal{X}$ is a feasible solution to $\Pi$.
It remains to lower bound the weight of that solution in terms of $w(\Xfeas)$. Note that
\[w(\Xhat) = w(\Xfeas) + w(\Xhat \setminus \Xfeas) - w(\Xfeas \setminus \Xhat).\]
Moreover, $[n] \setminus \Xfeas$
is a feasible solution to the complementary problem with weights
$\bar w$, because $\Xfeas \in \Xc$.
Hence, since $A$ is a $\rho$-approximation algorithm, we get
\[ \bar w(Y) \leqslant \rho \cdot \bar w\big([n] \setminus \Xfeas\big) = \rho \cdot w(\Xhat \setminus \Xfeas).
\]
Finally, we have
\begin{align*}
w(X) = w\big([n] \setminus Y\big)
&\geqslant w(\Xhat) - \bar{w}(Y)\\
&\geqslant w(\Xfeas) + w(\Xhat \setminus \Xfeas) - w(\Xfeas \setminus \Xhat)
    - \rho w(\Xhat \setminus \Xfeas)\\
&= w(\Xfeas) - (\rho-1)\eta^+ - \eta^-,
\end{align*}
where $(\eta^+,\eta^-)$ is the error of $\Xhat$ with respect
to $\Xfeas$.
If $\Xfeas$ is an optimal solution, we have
\begin{align*}
w(X)
&\geqslant \OPT - (\rho-1)\eta^+ - \eta^-
= \bigg(1 - \frac{(\rho - 1) \cdot \eta^+ + \eta^-}{\OPT}\bigg) \cdot \OPT.
\qedhere
\end{align*}
\end{proof}

Similarly to Theorem~\ref{thm:blackboxmin}, the algorithm implied by Theorem~\ref{thm:blackboxmax}
can be robustified by running it in parallel with
a conventional approximation algorithm $A'$ for problem $\Pi$.

\begin{corollary}
Under the same assumptions as in Theorem~\ref{thm:blackboxmax} there exists a learning augmented algorithm for $\Pi$ running in time $O(T(n))$ with approximation ratio
\[\max\left\{1-\frac{(\rho-1)\eta^+ +  \eta^-}{\OPT}, \rho' \right\},\]
where $\rho'$ is the approximation ratio of the best known
algorithm for $\Pi$ running in time $O(T(n))$.
\end{corollary}

\subsection{Example applications}
\label{sec:maximization-applications}

\paragraph{Maximum (and Max-Weight) Clique (and Independent Set).}
Given an undirected graph $G=(V,E)$, an \emph{independent set} (also called \emph{stable set}) is a subset of vertices $X \subseteq V$, such that no two vertices in $X$ share an edge. The Maximum Independent Set problem asks for an independent set of largest cardinality. In the Max-Weight Independent Set problem the input also includes vertex weights $w : V \to \Rplus$ and the goal is to find an independent set $X$ maximizing the total weight $w(X)$.
In the complement graph $G' = \smash{\big(V, \tbinom{V}{2} \setminus E\big)}$, these problems are equivalent to the Maximum (and Max-Weight) Clique problems, which ask for the largest cardinality (weight) complete subgraph $K_{\ell} \subseteq G$.

Maximum Clique (and Independent Set) is not only NP-hard to solve exactly~\citep{Karp72}, but it is also NP-hard to approximate within any factor better than $n^{1-\epsilon}$, for any $\epsilon > 0$~\citep{Hastad99}. Quite conveniently for us, the complementary problem to Max-Weight Independent Set is Min-Weight Vertex Cover, which can be $2$-approximated in linear time~\citep{Bar-YehudaE81}. Thus, applying Theorem~\ref{thm:blackboxmax},
we get a linear-time learning-augmented algorithm for Min-Weight Independent Set (and Clique) with approximation ratio
$1 - (\eta^+ + \eta^-)/\OPT$,
which is in striking contrast to the aforementioned impossibility of any nontrivial approximation ratio for conventional algorithms.

\paragraph{Knapsack.}
Given the knapsack capacity $c$ and $n$ items, the $i$-th of which has size $s_i$ and is worth~$w_i$, the Knapsack problem asks to find a subset of items of total size at most $c$ that maximizes the total worth.
Knapsack is another example from Karp's list of NP-hard problems~\citep{Karp72}, but it admits approximation schemes. In particular, it can be approximated to within factor $(1-\epsilon)$ in time $\smash{\tilde{O}(n+\epsilon^{-2})}$ \citep{CLMZ24,Mao24}. Unfortunately, these asymptotically optimal algorithms are based on structural results from additive combinatorics that make the constants hidden in the asymptotic notation enormous and render the algorithms themselves impractical. Perhaps a more practical approach is the $O(n \log (\epsilon^{-1}) + \epsilon^{-2.5})$-time algorithm by \citet{Chan18}. Still, if our goal is to solve with high accuracy (i.e., small $\epsilon$) many similar Knapsack instances, it might be a useful strategy to use Chan's algorithm only on a small fraction of those instances, learn from the obtained solutions a predicted solution, and feed it to a much faster learning-augmented algorithm in order to solve the remaining majority of instances.

In the following lemma we show that the problem complementary to Knapsack admits an $O(n \log n)$-time $2$-approximation algorithm. Theorem~\ref{thm:blackboxmax} then implies the existence of an $O(n \log n)$-time learning-augmented algorithm for Knapsack with approximation ratio
$1 - (\eta^+ + \eta^-)/\OPT.$

\begin{lemma}
There is an $O(n \log n)$-time $2$-approximation algorithm for the following problem: Given $n$ items, the $i$-th of which has size $s_i$ and is worth $w_i$, and the target $t$, find a subset of items of total size at least $t$ that minimizes the total worth.
\end{lemma}

\begin{algorithm2e}
\caption{Approximation algorithm for the problem complementary to Knapsack}
\label{alg:knapsack}
Sort items in the increasing order of $\frac{w_i}{s_i}$\;
$X := \emptyset$\;
$\mathcal{C} := \emptyset$\;
\For{$i = 1, \ldots, n$}{
  \eIf {$s(X) + s_i < t$}{
    $X := X \cup \{i\}$\;
  }{
    $\mathcal{C} := \mathcal{C} \cup \{X\}$\;
  }
}
\Return{$\operatorname{arg\,min}\{w(Y) \mid Y \in \mathcal{C}\}$}\;
\end{algorithm2e}

\begin{proof}
The algorithm (see Algorithm~\ref{alg:knapsack}) maintains an initially empty partial solution $X$, and keeps adding items to it in a greedy manner in the increasing order of the worth-to-size ratio $\frac{w_i}{s_i}$. Whenever adding the next item $i$ to $X$ would make the total size of $X$ meet or exceed the target size $t$, the algorithm skips adding that item to $X$ and instead adds $X\cup\{i\}$ to the set of candidate solutions $\mathcal{C}$. At the end, the algorithm returns the cheapest solution out of the candidate solutions $\mathcal{C}$.

To show that this is indeed a $2$-approximation algorithm, let $\OPT \subseteq [n]$ denote an optimal solution, and for every $i \in \{0,1,\ldots,n\}$ let $X_i$ denote the set $X$ after the algorithm considered the first $i$ items. It must be that $\OPT \setminus X_n \neq \emptyset$, because $s(X_n) < t \leqslant s(\OPT)$. Consider $i = \min(\OPT \setminus X_n)$, recalling that the items are numbered in an increasing order of the worth-to-size ratios. Since $i \notin X_i \subseteq X_n$, it must hold that $s(X_{i-1}) + s_i \geqslant t$, and thus $X_{i-1}\cup\{i\} \in \mathcal{C}$. By the minimality of $i$, we know that $\OPT \cap \{1, \ldots, i-1\} \subseteq X_{i-1}$, and thus each item in $X_{i-1} \setminus \OPT$ has smaller worth-to-size ratio than any item in $\OPT \setminus X_{i-1}$. This, together with the fact that $s(X_{i-1}) < t \leqslant s(\OPT)$ and hence also $s(X_{i-1} \setminus \OPT) < s(\OPT \setminus X_{i-1})$, implies that $w(X_{i-1} \setminus \OPT) < w(\OPT \setminus X_{i-1})$ and as a consequence $w(X_{i-1}) < w(\OPT)$. Clearly also $w_i \leqslant w(\OPT)$, so $\min\{w(Y) \mid Y \in \mathcal{C}\} \leqslant w(X_{i-1} \cup \{i\}) = w(X_{i-1}) + w_i < 2 \cdot \OPT$.
\end{proof}

\section{Refined bounds for Steiner Tree}
\label{sec:steiner-tree}

We describe a refined algorithm for Steiner Tree
based on a $2$-approximation algorithm by \citet{Mehlhorn88}.
Our careful analysis describes how it
detects and avoids false positives with high weight.
Our algorithm uses a parameter $\alpha \geqslant 1$
to adapt its treatment of the prediction.
With $\alpha=1$ its behavior copies Mehlhorn's algorithm,
and $\alpha$ approaching infinity
corresponds to the algorithm in Theorem~\ref{thm:blackboxmin}.
However, a different choice of $\alpha$ may give better
results.
We illustrate this by an example and
show how to achieve performance close to the best choice of $\alpha$ with only
a constant factor increase in running time in Section~\ref{sec:steiner-tree-alpha}.

The algorithm receives as input a graph $G = (V,E)$,
set $T\subseteq V$ of $k$ terminals,
a weight function $w: E \rightarrow \mathbb{R}_{\geqslant 0}$,
and a set $\smash{\Xhat \subseteq E}$ of predicted edges.
First, it scales down the weight of edges in $\smash{\Xhat}$ by dividing them by the parameter $\alpha$.
Then, it uses the algorithm of \citet{Mehlhorn88}
to compute a minimum spanning tree $\MST_\alpha$
of the metric closure
of the terminals with respect to the scaled edge weights.
In the end, the algorithm outputs the union of edges
contained in paths corresponding to the connections in
$\MST_\alpha$. 
The algorithm is summarized in Algorithm~\ref{alg:steiner}.    

\begin{proposition}[\citet{Mehlhorn88}]
\label{prop:mehlhorn}
The minimum spanning tree of the metric closure $\smash{\big(T, \tbinom{T}{2}\big)}$
of the graph $G$ can be computed
in time near-linear in the size of $G$.
\end{proposition}

\begin{algorithm2e}
\caption{Steiner tree with predictions}
\label{alg:steiner}
\textbf{Parameter:} $\alpha \geqslant 1$\;

\ForEach{$e\in E$}{
    \lIf{$e \in \Xhat$}{$w_{\alpha}(e) := w(e)/\alpha$}
    \lElse{$w_\alpha(e) := w(e)$}
}
Compute $\MST_\alpha$ of the metric closure $\mG = (T,\binom{T}{2})$
of $G$ with weights $w_\alpha$ using Proposition~\ref{prop:mehlhorn}\;
$X:= \emptyset$\;
\ForEach{{\rm edge} $e=\{t_1,t_2\}$ {\rm\bf in} $\MST_\alpha$}{
    Choose $p(e)\subseteq E$ the cheapest path in $G$ from
    $t_1$ to $t_2$ with respect to $w_\alpha$\;
    $X := X \cup p(e)$\;
}
\Return{$X$}\;
\end{algorithm2e}

\subsection{Analysis}
Recall that $G = (V, E)$ denotes the input graph, and let $\mG = (T, \binom{T}{2})$ denote the complete graph on the terminals $T\subseteq V$.
For any weight function $w: E \to \mathbb{R}_{\geqslant 0}$,
consider the shortest path metric on terminals induced by $w$,
i.e., for an edge $e$ in $\mG$ between $t_1$ and $t_2$, its cost $\mc(e)$ is equal to the length of the shortest path $p(e)$, with respect to $w$, from $t_1$ to $t_2$ in $G$.
There is a natural correspondence between an edge $e$ (the cost $\mc(e)$) in $\mG$ and the set of edges $p(e) \subseteq E$ in $G$ (weight $w(p(e))$).
Let $\MST$ denote a minimum spanning tree on $\mG$.
Our algorithm satisfies the following performance
bound.

\begin{theorem}
\label{thm:steiner}
Consider a graph $G=(V,E)$ with edge weights $w\colon E \to \R_{\geqslant0}$ and a set $T\subseteq V$ of $k$ terminals.
Let $\Xfeas \subseteq E$ be any Steiner tree on $G$ and
$S$ be a set of $\smash{\min\{k-1, |\Xhat\setminus \Xfeas|\}}$
edges with the highest cost contained in $\MST$.
Our refined algorithm outputs a Steiner tree $X$ of total weight
\begin{equation}
\label{eq:steiner}
w(X) \leqslant\left(1+\frac1\alpha\right)w(\Xfeas) + \left(1-\frac1\alpha\right)\eta^-
    +  \min\bigg\{\eta^+, (\alpha-1) \cdot
        \sum_{e\in S}\mc(e)
    \bigg\},
\end{equation}
where $(\eta^+, \eta^-)$ is the prediction error of $\Xhat$
with respect to $\Xfeas$.
\end{theorem}

A crucial property of Algorithm~\ref{alg:steiner} is that it never buys an edge
with weight more than $\alpha$ times larger than some connection in $\MST$.
The sum in \eqref{eq:steiner} needs to be over connections in
$\MST$ instead of individual edges in $\Xfeas$, since
$\Xfeas$ may consist of paths containing large number of edges of very small length.

With $\Xfeas = \Xstar$ being an optimal solution of cost $\OPT$ and
$\alpha=1$, the bound in \eqref{eq:steiner}
is equal to $2\OPT$, which corresponds
to the conventional algorithm which ignores the predictions.
With $\alpha$ approaching infinity, its limit is
$\OPT + \eta^+ + \eta^-$,
matching the result from Theorem~\ref{thm:blackboxmin}.
However, it can be much better than both in case of $\smash{\Xhat\setminus \Xstar}$
containing a small number of edges of very high weight.
Compared to Theorem~\ref{thm:blackboxmin}, $\eta^+$ in \eqref{eq:steiner}
is capped by $(\alpha-1)\sum_{e\in S}\mc(e)$  where
$\sum_{e\in S}\mc(e) \leqslant\mc(\MST) \leqslant2w(\Xstar)$, since Mehlhorn's algorithm is a $2$-approximation algorithm
for Minimum Steiner Tree.
Since we always have $\eta^- \leqslant w(\Xstar)$,
\eqref{eq:steiner} shows that Algorithm~\ref{alg:steiner}
is a $2\alpha$-approximation algorithm regardless of the prediction error.

The key part of the proof of Theorem~\ref{thm:steiner} is the analysis
of how edges in $\Xhat\cap \Xstar$ and edges in $\smash{\Xhat\setminus \Xstar}$ respectively
influence $\MST_\alpha$ found by Algorithm~\ref{alg:steiner}, depending
on the parameter $\alpha$. The improvement in the approximation ratio then
comes from the short-cutting procedure on the graph with scaled weights
which identifies paths over edges which are useful for connecting a higher number of
terminals.

Consider a fixed Steiner tree $\Xfeas \subseteq E$.
For the purpose of analysis, we define another weight function~$w_\alpha'$.
We set
    \[ w_\alpha'(e) := \begin{cases} 
      w(e)/\alpha & \text{if } e \in \Xhat \cap \Xfeas, \\
      w(e) & \text{otherwise}.
   \end{cases}
    \]
Denote by $\MST_\alpha'$ a minimum spanning tree in $\mG$ with respect to $\mc_\alpha'$, i.e., the metric closure of $w_\alpha'$.
The following observation holds.

\begin{observation}
The cost of the connections in $\MST_\alpha'$ can be bounded as
\begin{equation}\label{ineq:discounted_mst}
\mc_{\alpha}'(\MST_\alpha')
\leqslant 2w(\Xfeas) - 2(1-1/\alpha) w(\Xhat \cap \Xfeas).
\end{equation}
\end{observation}
\begin{proof}
In the graph $(V, \Xfeas)$,
replace each edge $uv \in \Xfeas$ by two directed edges
$(u,v)$ and $(v,u)$, each of weight $w(uv)$.
This way, each vertex has the same in-degree as out-degree and hence,
by Euler's theorem \citep[Theorem 2.24]{KV12},
there is a tour $P$ using each directed edge exactly once.
Since $\Xfeas$ is a Steiner tree, this tour visits all the terminals and naturally defines a spanning tree
on $\mG$ of cost at most $w(P)$: Whenever $P$ visits a new terminal
$t$, we add an edge $tt'\in \mG$ of cost
$\mc(tt') = w(p(tt'))$,
where $t'$ is the preceding terminal visited by $P$ and
$p(tt')$ is the segment of $P$ connecting $t'$ and $t$.

The weight of the tour $P$ is $2w(\Xfeas)$ with respect to $w$ and
$2w(\Xfeas) - (1-1/\alpha)w(\Xhat \cap \Xfeas)$
with respect to $w_\alpha'$, since every edge
$e\in \Xhat \cap \Xfeas$ is used twice and has cost
$w_\alpha'(e) = w(e) - (w(e) - w(e)/\alpha)$.
Therefore, we have
$\mc(\MST_\alpha') \leqslant 2w(\Xfeas) - (1-1/\alpha)w(\Xhat \cap \Xfeas)$.
\end{proof}

The following basic fact about spanning trees follows from the exchange
property of matroids
which states that for any two spanning trees $T_1$ and $T_2$
and any $e\in T_1 \setminus T_2$,
there is $e'\in T_2 \setminus T_1$ such that
$(T_1 \setminus \{e\}) \cup \{e'\}$ is a spanning tree.
For a proof of the exchange property, see for instance
\citet[Theorem~14.7]{KV12}.

\begin{proposition}
\label{prop:exchange}
Consider a minimum spanning tree $T$ on a graph $G$
with cost function $c\colon E \to \Rplus$
and an arbitrary spanning tree $T'$ on the same graph.
There exists a bijection $\phi \colon T \to T'$ such that
$c(e) \leqslant c(\phi(e))$ for each edge $e\in T$.
\end{proposition}

Proof of the following lemma contains the key part of our analysis.
It uses Proposition~\ref{prop:exchange} to charge
each edge in $\Xhat \setminus X'$ to a single connection
in $\MST_\alpha'$.

\begin{lemma}\label{lem:claim}
Let $S'$ be a set of $\min\{k-1, |\Xhat \setminus \Xfeas|\}$
edges with highest cost $\mc_\alpha'$ in $\MST_\alpha'$.
We have
\[w_\alpha'(X) \leqslant \mc_\alpha'(\MST_\alpha')
    + \min\bigg\{\eta^+, (\alpha-1)
        \sum_{e'\in S'}\mc_\alpha'(e')
    \bigg\}.\]
\end{lemma}
\begin{proof}
Denote by $e_1', e_2', \ldots, e_{k-1}'$ the edges of $\MST'_\alpha$ and by $e_1, e_2, \ldots, e_{k-1}$ the edges of $\MST_\alpha$. Since $\MST_\alpha$ is a minimum spanning tree with respect to $c_\alpha(\cdot)$ and up to reordering the edges,
by Proposition~\ref{prop:exchange}, we can assume that
\[c_\alpha(e_i) \leqslant c_\alpha(e_i'), \quad \forall i \in \{1, \ldots, k-1\}.\]

Let $F_{\leqslant 0} := \emptyset$ and set $F_{\leqslant i} := F_{\leqslant i-1} \cup (p(e_i)\setminus F_{\leqslant i-1})$, for $i = \{1, \ldots, k-1\}$, so that $F = F_{k-1}$. 

To show the lemma, for each pair $e_i, e_i'$, we distinguish between two cases:

\begin{enumerate}[nosep,left=0pt]
\item $\boldsymbol{(p(e_i)\setminus F_{\leqslant i-1})\cap(\hat{F}\setminus F^*) = \emptyset}$. In this case, no weights of edges on $(p(e_i)\setminus F_{i-1}) \cap (\hat{F}\setminus F^*)$ have been scaled down, hence
\[c_\alpha'(p(e_i)\setminus F_{\leqslant i-1}) = c_\alpha(p(e_i)\setminus F_{i-1}) \leqslant c_\alpha(e_i') \leqslant c_\alpha'(e_i').\]
\item $\boldsymbol{(p(e_i)\setminus F_{\leqslant i-1})\cap(\hat{F}\setminus F^*) \neq \emptyset}$. In this case, the actual cost of the path corresponding to edge $e_i$ minus the edges in $F_{\leqslant i-1}$ already bought, i.e. $p(e_i)\setminus F_{\leqslant i-1}$, can be at most $\alpha$ times higher:
\[c_\alpha'(p(e_i)\setminus F_{\leqslant i-1}) \leqslant \alpha \cdot c_\alpha(e_i) \leqslant \alpha\cdot c_\alpha(e_i') \leqslant c_\alpha'(e_i') + (\alpha-1) \cdot c_\alpha'(e_i').\]
At the same time, we have
\[ c_\alpha'(p(e_i)\setminus F_{\leqslant i-1}) \leq
        c_\alpha(e_i) + c(\hat F \cap (p(e_i)\setminus F_{\leqslant i-1}))
        \leqslant c_\alpha'(e_i') + c(\hat F \cap (p(e_i)\setminus F_{\leqslant i-1})).
\]
\end{enumerate}
Since the latter case only happens once for each edge in $\hat{F}\setminus F^*$, we see that
\[
c_\alpha'(F) \leqslant \underbrace{\sum_{i=1}^{k-1} c_\alpha'(e_i')}_{c_\alpha'(\MST_\alpha')} + (\alpha-1)\max_{S\subseteq [k]: |S| = |\hat{F}\setminus F^*|}\sum_{i\in S} c_\alpha'(e_i'),
\]
and at the same time, we have $ c_\alpha'(F) \leqslant \MST_\alpha' + c(\hat F).$
The lemma follows.
\end{proof}

\begin{proof}[Proof of Theorem~\ref{thm:steiner}]
Since for all edges $e\in \Xhat \cap \Xfeas$ we have
$w_\alpha'(e) = w(e)/\alpha$, we can write
$w(e) = w_\alpha'(e) + (1-1/\alpha)w(e)$. Therefore, we have
\[w(X) \leqslant w_\alpha'(X) + (1-1/\alpha) \cdot w(\Xhat \cap \Xfeas).\]
Combining this bound with Lemma~\ref{lem:claim} and
\eqref{ineq:discounted_mst}, we get
\[ w(X) \leqslant 2 w(\Xfeas) - (1-1/\alpha) w(\Xhat \cap \Xfeas)
    + \min\bigg\{\eta^+, (\alpha-1)
        \sum_{e'\in S'}\mc_\alpha'(e')
    \bigg\}.
\]
Proposition~\ref{prop:exchange}
implies that
$\sum_{e'\in S'} c_\alpha'(e') \leq
\sum_{e\in S} c(e)$, since
$\MST_\alpha'$ is a minimum spanning tree on $\mG$
with respect to cost $c_\alpha'$ and
$\MST$ is a spanning tree on the same graph.
Now, it is enough to note that
\[ 2w(\Xfeas) - (1-1/\alpha) w(\Xhat \cap \Xfeas)
= (1+1/\alpha) w(\Xfeas) + (1-1/\alpha)w(\Xfeas\setminus \Xhat).
\qedhere
\]
\end{proof}

\subsection{Approximating the best alpha}
\label{sec:steiner-tree-alpha}

Let $\alpha^*$ be the parameter of Algorithm~\ref{alg:steiner}
which minimizes the upper bound in Theorem~\ref{thm:steiner}.
We show how to find $\alpha$ which achieves a bound at
most $(1+\epsilon)$ times worse than the bound with $\alpha^*$.

\begin{algorithm2e}
\caption{Steiner tree: search for the best $\alpha$}
\label{alg:steiner-alpha}
\For{$i = 0, \dotsc, \lceil \log_{1+\epsilon} \epsilon^{-1}\rceil$}{
    run Algorithm \ref{alg:steiner} with $\alpha = (1+\epsilon)^i$\;
}
output the best solution found during above iterations\;
\end{algorithm2e}

\begin{theorem}
\label{thm:steiner-alpha}
For a constant $\epsilon > 0$, Algorithm~\ref{alg:steiner-alpha}
runs in near-linear time and
finds a Steiner tree $X$ with weight
\[
w(X) \leqslant (1+\epsilon) \min_{\alpha\geqslant 1} \bigg\{
    \left(1+\frac1\alpha\right)w(\Xfeas) + \left(1-\frac1\alpha\right)\eta^-
    +  \min\bigg\{\eta^+, (\alpha-1) \cdot
        \sum_{e\in S}\mc(e)
    \bigg\}
\bigg\}.
\]
\end{theorem}
\begin{proof}
Algorithm~\ref{alg:steiner-alpha} performs a constant number
$\lceil\log_{1+\epsilon} \epsilon^{-1}\rceil + 1$ of iterations
of Algorithm~\ref{alg:steiner}, which runs in near-linear time.

First, note that we need to consider only $\alpha \leqslant \epsilon^{-1}$,
because
\begin{align*}
(1+\epsilon)c(\Xfeas) + (1-\epsilon)\eta^-
    +  \min\bigg\{\eta^+, (\epsilon^{-1} -1) \cdot
    \sum_{e\in S}\mc(e)\bigg\}\bigg)
    \leqslant (1+\epsilon)f(\alpha)
\end{align*}
for any $\alpha > \epsilon^{-1}$.

It is enough to show that, for any $\alpha$, we have
$f((1+\epsilon)\alpha) \in [(1+5\epsilon)^{-1}f(\alpha), (1+5\epsilon)f(\alpha)]$.
We show this for every term of $f(\alpha)$ separately.
We have
\[
\big(1+\frac1{\alpha}\big)c(\Xfeas)
    \geqslant \bigg(1+\frac{1}{(1+\epsilon)\alpha}\bigg)c(\Xfeas)
    = \frac{(1+\epsilon)\alpha + 1}{(1+\epsilon)\alpha}c(\Xfeas)
    \geqslant (1+\epsilon)^{-1} (1+\frac{1}{\alpha})c(\Xfeas).
\]
Similarly, we have
\[
\big(1-\frac{1}{\alpha}\big)\eta^-
    \leqslant \bigg(1-\frac{1}{(1+\epsilon)\alpha}\bigg)\eta^-
    \leqslant \bigg(\frac{(1+\epsilon)\alpha - 1}{(1+\epsilon)\alpha}\bigg)\eta^-
    \leqslant (1+\epsilon)\bigg(\frac{\alpha - 1}{(1+\epsilon)\alpha}\bigg)\eta^-,
\]
which is at most $(1+\epsilon)(1-1/\alpha)\eta^-$.
For the last term, we have
\[ \alpha-1 \leqslant (1+\epsilon)\alpha -1
    = \alpha - 1 + \epsilon\alpha
    \leqslant (1+2\epsilon)(\alpha - 1)
\]
if $\alpha \geqslant 2$.
If $\alpha < 2$, we have $\epsilon\alpha
    \sum_{e\in S}\mc(e_i)
    \leqslant 4\epsilon c(\Xfeas)$.
\end{proof}

\subsection{Tight example for our analysis}
\label{sec:steiner-tight}

Figure~\ref{fig:steiner-lb} describes an instance with $k=n-1$ terminals
and a prediction
with a single false-negative edge of weight $1+\epsilon$
and a single false-positive edge of weight $\beta > 2$.
Algorithm~\ref{alg:steiner} with $\alpha \in [1, (1+\epsilon)^{-1})$
achieves approximation ratio approaching $2$ as $n$ increases.
With $\alpha \in (\frac{1}{1+\epsilon}, \frac{\beta}{1+\epsilon})$,
its approximation ratio is equal to 1.
With $\alpha > \frac{\beta}{1+\epsilon}$,
it approaches
$\frac{\OPT + \beta}{\OPT}$, which is equal to $2$ if we choose $\beta = \OPT$.
This shows that the best choice of $\alpha$ may not be $1$ nor approaching $\infty$.

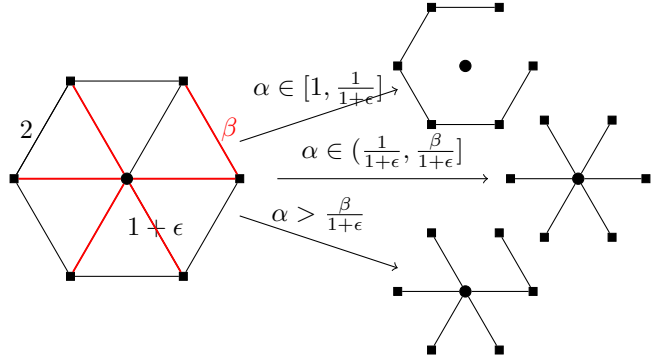
\begin{figure}[b]
    \centering
    \begin{tikzpicture}[scale=1.5]
    \begin{scope}
    \pgfmathtruncatemacro{\k}{6} 
    
    \foreach \i in {1,...,\k}
        \draw (\i*360/\k:1cm) node[rectangle, draw, fill=black, inner sep=1.5pt] (v\i) {};
    
    \node[circle, draw, fill=black, inner sep=1.5pt] (v\the\numexpr\k+1\relax) at (0,0) {};
    
    \foreach \i [evaluate=\i as \j using {int(mod(\i,\k)+1)}] in {1,...,\k}
        \draw (v\i) -- (v\j);
    
    \foreach \i in {1,...,\k}
        \draw (v\the\numexpr\k+1\relax) -- (v\i);

    \draw[line width = 0.25mm, red] (v1) -- (v\k) node[midway, right] {$\beta$};

    \draw (v2) -- (v3) node[midway, left] {$2$};

    \draw (v\the\numexpr\k+1\relax) -- (v5) node[pos = 0.5] {$1+\epsilon$};

    \foreach \i in {2,...,\k}
        \draw[line width = 0.25mm, red] (v\the\numexpr\k+1\relax) -- (v\i);

    \draw[->] (1,0.33) -- (1.2*2,1.2*0.66) node[midway, above] {$\alpha \in [1, \tfrac{1}{1+\epsilon}]$};
    \draw[->] (4/3,0) -- (1.2*8/3,1.2*0) node[midway, above] {$\alpha \in (\tfrac{1}{1+\epsilon}, \tfrac{\beta}{1+\epsilon}]$};
    \draw[->] (1,-0.33) -- (1.2*2,-1.2*0.66) node[midway, above] {$\alpha > \tfrac{\beta}{1+\epsilon}$};

    \end{scope}

    \begin{scope}[xshift = 3cm, yshift = 1cm, scale = 0.6]
    \pgfmathtruncatemacro{\k}{6} 
    
    \foreach \i in {1,...,\k}
        \draw (\i*360/\k:1cm) node[rectangle, draw, fill=black, inner sep=1.5pt] (v\i) {};
    
    \node[circle, draw, fill=black, inner sep=1.5pt] (v\the\numexpr\k+1\relax) at (0,0) {};
    
    \foreach \i [evaluate=\i as \j using {int(mod(\i,\k)+1)}] in {1,...,\the\numexpr\k-1\relax}
        \draw (v\i) -- (v\j);
    \end{scope}

    \begin{scope}[xshift = 4cm, yshift = 0cm, scale = 0.6]
    \pgfmathtruncatemacro{\k}{6} 
    
    \foreach \i in {1,...,\k}
        \draw (\i*360/\k:1cm) node[rectangle, draw, fill=black, inner sep=1.5pt] (v\i) {};
    
    \node[circle, draw, fill=black, inner sep=1.5pt] (v\the\numexpr\k+1\relax) at (0,0) {};
    
    \foreach \i in {1,...,\k}
        \draw (v\the\numexpr\k+1\relax) -- (v\i);
    \end{scope}

    \begin{scope}[xshift = 3cm, yshift = -1cm, scale = 0.6]
    \pgfmathtruncatemacro{\k}{6} 
    
    \foreach \i in {1,...,\k}
        \draw (\i*360/\k:1cm) node[rectangle, draw, fill=black, inner sep=1.5pt] (v\i) {};
    
    \node[circle, draw, fill=black, inner sep=1.5pt] (v\the\numexpr\k+1\relax) at (0,0) {};
    
    \foreach \i in {2,...,\k}
        \draw (v\the\numexpr\k+1\relax) -- (v\i);

    \draw (v1) -- (v\k);
    \end{scope}

\end{tikzpicture}
    \caption{Steiner tree problem with $k=n-1$ terminals (square) and one non-terminal vertex (circle). Outer edges have weight $2$, except one edge of weight $\beta$. Inner edges have weight $1+\epsilon$. Red edges are predicted. On the right are the respective Steiner trees output by Algorithm~\ref{alg:steiner} with different parameters of $\alpha$.}
    \label{fig:steiner-lb}
\end{figure}

We can use this construction to show that the coefficient of $\sum_{e\in S} \mc(e)$
in \eqref{eq:steiner} is tight for the given algorithm.
First consider the input graph in Figure~\ref{fig:steiner-lb}
with $\beta = 2n$. Algorithm~\ref{alg:steiner}
with $\alpha = \beta$ receiving a prediction $\Xhat$ containing
all red edges, achieves cost
$(n-1)(1+\epsilon) + \beta \geqslant \OPT + \frac\alpha2 2 - (1+\epsilon)$,
where $2$ is the cost of the most expensive connection in
the minimum spanning tree on the metric closure of the terminals.
Note that all terms except for $-(1+\epsilon)$ go to infinity as $n$ increases.

Similarly, we can show the tightness of the coefficient of $\eta^-$.
Algorithm~\ref{alg:steiner} with $\Xhat = \emptyset$
achieves cost $2n \geqslant (1+\epsilon)^{-1} \big((1+1/\alpha) \OPT + \eta^-$,
since $\eta^- = \OPT = n(1+\epsilon)$.

\clearpage

\section{Lower bounds}
\label{sec:lower-bounds}

\subsection{Preliminaries}

In order to obtain lower bounds for our setting we use the following results of~\citet{KhotR08}.

\begin{proposition}[\cite{KhotR08}]
\label{prop:IS}
Assuming UGC, for every constant $\delta>0$, it is NP-hard to distinguish, for an input $n$-vertex graph $G$, between the following two cases:
\begin{itemize}
    \item (YES) $G$ contains an independent set of size at least $\frac{n}{2}  -\delta n$.
    \item (NO) $G$ contains no independent set of size at least $\delta n$.
\end{itemize}
\end{proposition}

It directly implies the following corollary for the Minimum Vertex Cover problem:

\begin{corollary}[\cite{KhotR08}]
\label{cor:VC}
Assuming UGC, for every constant $\delta>0$ it is NP-hard to distinguish, for an input $n$-vertex graph $G$, between the following two cases:
\begin{itemize}
    \item (YES) $G$ has a vertex cover of size at most $\frac{n}{2}+\delta n$.
    \item (NO) $G$ has no vertex cover of size at most $(1-\delta) n$.
\end{itemize}
\end{corollary}

Note that Corollary~\ref{cor:VC} implies that, assuming UGC and P$\neq$NP, Minimum Vertex Cover cannot be approximated within a factor of $2-\epsilon$, for any $\epsilon > 0$.

\subsection{Lower bound for a minimization problem}

We now argue that our Theorem~\ref{thm:blackboxmin} for minimization selection problems cannot be improved.
For a specific minimization selection problem, namely the Minimum Vertex Cover problem, which has a folklore 2-approximation
algorithm, Theorem~\ref{thm:blackboxmin} and Corollary~\ref{cor:robust} with $\rho=2$ imply existence
of a learning-augmented algorithm with approximation ratio
$\min\{1+ \smash{\frac{\eta^+ + \eta^-}{\OPT}},2\}$. We show that this is the best possible upper bound.

\begin{theorem}
Let $A$ be a polynomial-time
learning-augmented approximation algorithm for Minimum Vertex Cover that guarantees an approximation ratio of at most
\[1+f\left(\frac{\eta^+}{\OPT},\frac{\eta^-}{\OPT}\right),\]
for some non-decreasing function $f$. Then, assuming UGC, we have
\[f(x, y) \geqslant x + y,\]
for any choice of $x, y \geqslant 0$ satisfying
$x + y \leqslant 1$.
\label{thm:lb-min}
\end{theorem}
\begin{proof}
By contradiction,
consider $x,y$ such that $f(x,y) < x + y$.
We denote $\epsilon := (x+y) - f(x,y)$, and let $\delta := \epsilon/20$.
We will show how to use the hypothesized learning-augmented algorithm $A$ to construct an algorithm $\bar A$
for the (prediction-less) Minimum Vertex Cover problem
that can distinguish between instances on $n$ vertices with optimum value more
than $(1-\delta)n$ and instances with optimum value
at most $(1+\delta)\frac{n}2$.
Such an algorithm would contradict Corollary~\ref{cor:VC}.
In order to provide a cleaner exposition, we first assume that $A$ can actually
solve \emph{Min-Weight} Vertex Cover with the given approximation ratio.
This assumption can be removed easily, as discussed at the end of the proof.

We construct algorithm $\bar A$ that, for any input graph
$G$ with minimum vertex cover $X^*$ of size
$|X^*| \in  \big[(1-\delta)\frac{n}2, (1+\delta)\frac{n}2\big]$,
produces a solution $X$ such that $|X| < (1-\delta)n$,
certifying that $G$ does not belong to the NO case.
Note that for graphs with vertex cover of size $< (1-\delta)\frac{n}2$
such solution can be easily produced by the folklore 2-approximation algorithm.

Given the input graph $G$,
$\bar A$ produces a weighted graph $\bar G$ and a
prediction $\hat X$, and runs $A$ on $\bar G$ with $\hat X$.
Graph $\bar G$ is a disjoint union of three graphs: $G_0$, $G_+$, and $G_-$. Let $x' := \frac{1-\delta}{1+\delta}x$
and $y' := \frac{1-\delta}{1+\delta}y$.
\begin{itemize}[left=0pt]
\item
$G_0$ is a complete graph on $\frac{n}2+1$ vertices, we have
$\bar w(v)=(1-x'-y')$ for each $v\in V(G_0)$.
\item $G_+$ is a copy of $G$ with scaled weights, we have
$\bar w(v) = x'$ for each $v\in V(G_+)$.
\item $G_-$ is a copy of $G$ with scaled weights, we have
$\bar w(v) = y'$ for each $v\in V(G_-)$.
\end{itemize}
The prediction
$\hat X$ contains arbitrary $\frac{n}2$ vertices from $G_0$,
all vertices of $V(G_+)$ and no vertices of $V(G_-)$.
Receiving a solution $\bar X$ found by the algorithm $A$ on $\bar G$,
we denote $X_+ := \bar X \cap V(G_+)$ and  $X_- := \bar X \cap V(G_-)$.
Algorithm $\bar A$ returns $X$ which is either $X_+$ or $X_-$, choosing
the one with the smaller size.
We now prove that $|X| < (1-\delta)n$ if $G$ is the YES case.

First, we show that $\eta^+/\OPT \leqslant x$ and
$\eta^-/\OPT \leqslant y$.
Let $\OPT$ be the value of an optimal solution for $\bar G$.
We have $\OPT = (1-x-y)\frac{n}2 + x|X^*| + y|X^*|$ which belongs to
$\big[(1-\delta)\frac{n}2, (1+\delta)\frac{n}2\big]$ by the assumption about
the size of $X^*$.
For the error, we have
\begin{align*}
\eta^+ &= x'\cdot (n-|X^*|)=\frac{1-\delta}{1+\delta}x\cdot (n-|X^*|),\quad\text{and}\\
\eta^- &= y'\cdot |X^*| = \frac{1-\delta}{1+\delta}y\cdot |X^*|.
\end{align*}
Since both $(n-|X^*|)$ and $|X^*|$ are at most $(1+\delta)\frac{n}{2}$
and $\OPT \geqslant (1-\delta)\frac{n}{2}$,
we have $\eta^+/\OPT \leqslant x$ and $\eta^-/\OPT \leqslant y$.
By monotonicity of $f$,
this implies that $\bar X$ is a $(1+x+y-\epsilon)$-approximate solution.
Therefore, we have
\begin{equation}
\label{eq:lb-min}
\begin{split}
\bar w(\bar X)
&\leqslant (1+x+y-\epsilon)\cdot\OPT\\
&\leqslant (1+x+y)\cdot(1+\delta)\frac{n}{2} - \epsilon\cdot (1-\delta)\frac{n}{2}\\
&\leqslant \frac{n}{2} + (x+y)\frac{n}{2} + (1+x+y+\epsilon)\delta \frac{n}{2} - \epsilon \frac{n}{2}\\
&\leqslant \frac{n}{2} + (x+y)\frac{n}{2} + 4\delta \frac{n}{2} - \epsilon \frac{n}{2}\\
&\leqslant \frac{n}{2} + (x+y)\frac{n}{2} - 0.8\epsilon \frac{n}{2},
\end{split}
\end{equation}
by our choice of $\delta = \epsilon/20$. On the other hand, we can write
\begin{equation}
\label{eq:lb-min2}
\begin{split}
\bar w(\bar X)
&\geqslant  (1-x'-y')\frac{n}{2} + x'|X_+| + y'|X_-|\\
&= \frac{n}{2} + x'\big(|X_+|- \frac{n}{2}\big) + y'\big(|X_-|-\frac{n}{2}\big)\\
&\geqslant \frac{n}{2} + x\big(|X_+|- \frac{n}{2}\big) + y\big(|X_-|-\frac{n}{2}\big) - 4\delta \frac{n}{2}\\
&\geqslant \frac{n}{2} + x\big(|X_+|- \frac{n}{2}\big) + y\big(|X_-|-\frac{n}{2}\big) - 0.2\epsilon \frac{n}{2}.
\end{split}
\end{equation}
In the second to last step, we use the definition of $x'$ and $y'$, and the fact that
$\frac{1-\delta}{1+\delta} \geqslant 1-2\delta$.
Inequalities \eqref{eq:lb-min} and \eqref{eq:lb-min2}
imply that at least one of $|X_+|-\frac{n}{2}$ and $|X_-|-\frac{n}{2}$
must be less than or equal to $\frac{n}{2} - 0.6\epsilon \frac{n}{2} < \frac{n}{2} - \delta n$. Indeed, otherwise we would have
\[
\bar w(\bar X)
> \frac{n}{2}+ (x+y)\frac{n}{2} - (x+y) 0.6\epsilon  \frac{n}{2} - 0.2\epsilon \frac{n}{2}
\geqslant \frac{n}{2}+ (x+y)\frac{n}{2}  - 0.8\epsilon  \frac{n}{2},
\]
contradicting \eqref{eq:lb-min}.
In other words, we have either $|X_+| < (1-\delta)n$ or $|X_-| < (1-\delta)n$,
which concludes the proof.

In order to transform $\bar G$ into an unweighted instance,
we take (roughly) $1/x'y'$ copies of $G_0$, $1/(1-x'-y')y'$ copies of $G_+$
and $1/(1-x'-y')x'$ copies of $G_-$, each vertex having weight 1.
Prediction contains $n/2$ vertices from each copy of $G_0$, all vertices
of all copies of $G_+$ and no vertices from the copies of $G_-$.
This is to ensure that the minimum vertex cover on the copies of $G_0$, on the
copies of $G_+$, and on the copies of $G_-$ respectively have
approximately the same ratio as in $\bar G$.
In order to achieve approximation ratio $1+x+y-\epsilon$,
the algorithm $A$ would need to find a solution to at least one copy
of $G_-$ or $G_+$ with size smaller than $(1-\delta)n$.
\end{proof}

\subsection{Lower bound for a maximization problem}

We show a similar result for the class of considered maximization problems.
\begin{theorem}
Let $A$ be a polynomial-time
learning-augmented approximation algorithm for Maximum Independent Set that guarantees an approximation ratio of at least
\[1-f\left(\frac{\eta^+}{\OPT},\frac{\eta^-}{\OPT}\right),\]
for some non-decreasing function $f$. Then, assuming UGC, we have
\[f(x, y) \geqslant x + y,\]
for any choice of
$x, y \geqslant 0$
satisfying
$x + y \leqslant 1$.
 \label{thm:lb-max}
\end{theorem}
\begin{proof}
By contradiction,
consider $x,y$ such that $f(x,y) < x + y$.
We denote $\epsilon := (x+y) - f(x,y)$, and let $\delta := \epsilon/20$.
We proceed similarly to the proof of Theorem~\ref{thm:lb-min},
and we reach contradiction with Proposition~\ref{prop:IS}.

We construct algorithm $\bar A$ such that for any input graph
$G$ with maximum independent set $X^*$ of size
$|X^*| \in  [(1-\delta)\frac{n}{2}, (1+\delta)\frac{n}{2}]$
it produces an independent set $X$ of size $|X| > \delta n$,
certifying that $G$ does not belong to the NO case.
Note that graphs with an independent set larger than $(1+\delta)\frac{n}{2}$
have a minimum vertex cover of size smaller than $(1-\delta)\frac{n}{2}$.
Therefore, in such graphs we can find a 2-approximate solution $C$ of the minimum vertex cover problem and
then $V(G)\setminus C$ will be an independent set of size $n - |C| > n - (1-\delta)n = \delta n$.

Given the input graph $G$, $\bar A$ constructs a (weighted) graph $\bar G$
and a prediction $\hat X$, 
and uses them as an input for algorithm $A$.
Graph $\bar G$ is a disjoint union of three graphs: $G_0$, $G_+$, and $G_-$. Let $x' := \frac{1-\delta}{1+\delta}x$ and $y' := \frac{1-\delta}{1+\delta}y$.
\begin{itemize}[left=0pt]
\item
$G_0$ is a graph with $n/2$ vertices and no edges, we have
$\bar w(v)=(1-x'-y')$ for each $v\in V(G_0)$.
\item $G_+$ is a copy of $G$ with scaled weights, we have
$\bar w(v) = x'$ for each $v\in V(G_+)$.
\item $G_-$ is a copy of $G$ with scaled weights, we have
$\bar w(v) = y'$ for each $v\in V(G_-)$.
\end{itemize}
The prediction
$\hat X$ contains all vertices from $G_0$,
all vertices of $V(G_+)$ and no vertices of $V(G_-)$.
Receiving a solution $\bar X$ found by algorithm $A$ on $\bar G$,
we denote $X_+ = \bar X \cap V(G_+)$ and  $X_- = \bar X \cap V(G_-)$.
$\bar{A}$ returns the larger of the two sets $X_+$ and $X_-$, which we denote by $X$.
We now prove that $|X| > \gamma n$ if $G$ is the YES case.

First, we show that $\eta^+/\OPT \leqslant x$ and
$\eta^-/\OPT \leqslant y$.
Let $\OPT$ be the value of an optimal solution on $\bar G$.
We have $\OPT = (1-x-y)\frac{n}2 + x|X^*| + y|X^*|$ which belongs to
$\big[(1-\delta)\frac{n}2, (1+\delta)\frac{n}2\big]$ by the assumption about
the size of $X^*$.
For the error, we have
\begin{align*}
\eta^+ &= x'\cdot (n-|X^*|)=\frac{1-\delta}{1+\delta}x\cdot (n-|X^*|),\quad\text{and}\\
\eta^- &= y'\cdot |X^*| = \frac{1-\delta}{1+\delta}y\cdot |X^*|.
\end{align*}
Since both $(n-|X^*|)$ and $|X^*|$ are at most $(1+\delta)\frac{n}{2}$
and $\OPT \geqslant (1-\delta)\frac{n}{2}$,
we have $\eta^+/\OPT \leqslant x$ and $\eta^-/\OPT \leqslant y$.
By monotonicity of $f$, this implies that $\bar X$ is at least a $(1-x-y+\epsilon)$-approximate solution. Therefore, we have
\begin{align*}
\bar w(\bar X)
&\geqslant (1-x-y+\epsilon)\cdot\OPT\\
&\geqslant (1-x-y)\cdot (1-\delta)\frac{n}{2} + \epsilon\cdot (1-\delta)\frac{n}{2}\\
&\geqslant (1-x-y)\frac{n}{2} - 3\delta \frac{n}{2} + \epsilon \frac{n}{2}.
\end{align*}
On the other hand, we can write
\begin{align*} \bar w(\bar X)
&\leqslant  (1-x'-y')\frac{n}{2} + x'|X_+| + y'|X_-|\\
&\leqslant (1-x-y)\frac{n}{2} + x'|X_+| + y'|X_-| + 4\delta\frac{n}{2}.
\end{align*}
In the last step, we use the definition of $x'$ and $y'$, and
$\frac{1-\delta}{1+\delta} \geqslant 1-2\delta$.
The two equations above imply that at least one of $|X_+|$ and $|X_-|$
must be larger than $(\epsilon \frac{n}{2} - 7\delta\frac{n}{2})/2 > \delta n$,
by our choice of $\delta$.
In other words, we have that
$|X_+| > \delta n$ or $|X_-| > \delta n$, which concludes the proof.

In order to transform $\bar G$ into an unweighted instance,
we take $\smash{\frac{1}{x'y'}}$ copies of $G_0$, $\smash{\frac{1}{(1-x'-y')y'}}$ copies of $G_+$
and $\smash{\frac{1}{(1-x'-y')x'}}$ copies of $G_-$, each vertex having weight 1.
\end{proof}

\section{Experimental evaluation}
\label{sec:experiments}

In this section we present an experimental evaluation of our refined learning-augmented algorithm for the Steiner Tree problem from Section~\ref{sec:steiner-tree}. The source code is available on GitHub\footnote{Available at \href{https://github.com/adampolak/steiner-tree-with-predictions}{\texttt{github.com/adampolak/steiner-tree-with-predictions}}.}.

\paragraph{Dataset.} We base our experiments on the heuristic track of the 2018 PACE Challenge~\citep{BonnetS18}, which is to our knowledge the most recent large-scale computational challenge concerning the Steiner Tree problem. In particular, we use their dataset\footnote{Available at \href{https://github.com/PACE-challenge/SteinerTree-PACE-2018-instances}{\texttt{github.com/PACE-challenge/SteinerTree-PACE-2018-instances}}.}, which consists of $199$ graphs selected from among the hardest instances in the SteinLib library~\citep{SteinLib}.\footnote{As of 2018, a majority of these instances could not be solved to optimality within one hour with state-of-the-art solvers. The average number of vertices is $\approx 27$k, the average number of edges is $\approx 48$k and the average number of terminals is $\approx 1$k.} 

\paragraph{Algorithms.} We also use PACE as the source of the state-of-the-art Steiner Tree solver, with which we compare our algorithm. More specifically, in our experiments we evaluate three algorithms:
\begin{itemize}[left=0pt]
\item \textbf{Mehlhorn}'s $2$-approximation algorithm, implemented by us in C++. Its mean empirical approximation factor on PACE instances that we observe is $\approx 1.17$. On average it spends $34$ milliseconds per instance, and never needs more than $500$ milliseconds.
\item \textbf{CIMAT}, the winning solver from the PACE Challenge, whose C++ implementation is available on GitHub~\citep{CIMAT}. It is an evolutionary algorithm that can be stopped at any time and outputs the best solution found so far. It was designed to run for $30$ minutes per instance, as this was the time limit in PACE. However, we noticed that on $95\%$ of instances already after one minute it produces a solution within $1.01$ of the optimum, and therefore we decided to always run it only for one minute, in order to keep the computational costs of the experiments down. We also remark that, on average, it took the CIMAT solver $\approx10$ seconds to output a first feasible solution. In all our experiments, we use the value of the solution returned by CIMAT as an estimate of the value of an optimal solution $\OPT$, which would be difficult and impractical to calculate exactly.
\item \textbf{ALPS}, our \textbf{al}gorithm with \textbf{p}rediction\textbf{s} from Section~\ref{sec:steiner-tree}, with the ``confidence'' hyperparameter $\alpha \in \{1.1, 1.2, 1.4, 2, 4, \infty\}$. In terms of efficiency it is virtually indistinguishable from Mehlhorn's algorithm, which it runs as a subroutine and which dominates its running time.
\end{itemize}

\paragraph{Predictions.}
We generate both synthetic predictions and learned predictions. For both types of predictions we vary their quality, which is captured by the parameter $p\in \{0.0, 0.1, \ldots, 1.0\}$; the higher the parameter $p$ the higher the prediction error. The predictions are then fed to our algorithm (ALPS), and we compare the quality of its outputs with those of the other two algorithms (CIMAT and Mehlhorn's).
\begin{itemize}[left=0pt]
    \item \textbf{Synthetic predictions.}
    In our first experiment, we simulate a predictor by introducing artificial noise to groundtruth labels. For each instance from the dataset, we compute a (near-)optimal solution $S\subseteq E$ using CIMAT. Then, for each value of the parameter $p$, we swap out a randomly selected $p$-fraction of the edges in $S$ for a randomly selected subset of edges from $E \setminus S$ of the same cardinality. This yields a prediction with error roughly $\eta^+ \approx p\cdot\OPT$ and $\eta^- \approx p\cdot\OPT$.
    \item \textbf{Learned predictions.}
    In our second, more realistic experiment we test our approach end-to-end, i.e., we first learn predictions, and then use them to solve instances not seen during learning. For this experiment, we need to work with \emph{distributions over instances} instead of individual problem instances. To this end, for each instance $I = (V, E, T, w)$ from our dataset, and for each value of the parameter $p$, we consider a distribution that returns instances of the form $I' = (V, E, T', w)$, with a $p$-fraction of the terminals resampled (but with the same underlying graph and edge weights). Actually, we consider two types of distributions, depending on how they resample terminals:
    \begin{itemize}
      \item \textbf{``Fixed core'' distribution.}
      We fix a $(1-p)$-fraction of the terminals $T_{\text{old}} \subseteq T$ and for each sampled instance $I'$ we swap out the remaining $p$-fraction by independently sampling $\lceil p \cdot |T| \rceil$ terminals $T_{\text{new}} \subseteq V \setminus T$, returning $I' = (V, E, T_{\text{old}} \cup T_{\text{new}}, w)$.
      \item \textbf{``No core'' distribution.}
      Each time we sample a fresh $(1-p)$-fraction of the terminals in $T$ to obtain $T_{\text{old}}$ (i.e., $T_{\text{old}} \subseteq T$ is not fixed over all the samples).
    \end{itemize}
    For each such distribution, we sample $k=10$ instances, and we compute a (near-)optimal solution to each of them using CIMAT. Then, we evaluate ALPS on each such sampled instance using the prediction learned from (the solutions to) the remaining $k-1=9$ instances sampled from the same distribution (i.e., we preform leave-one-out cross-validation).
    To learn the prediction from a set of left-out instances, we predict edge $e \in E$ iff it appeared in more than half of the (near-)optimal solutions to the left-out instances, which coincides with empirical risk minimization of the combined prediction error $\eta_+ + \eta_-$.
\end{itemize}

\paragraph{Evaluation metrics.} 
To evaluate the performance of our algorithms, we use two different measures.
For the synthetic predictions we look at the (empirical) approximation ratios observed for each of the algorithms and averaged over all the instances (see Figure~\ref{subfig:a}).
For the learned predictions such average would not be a useful metric -- this is because swapping out a set of terminals often completely changes both the value of an optimal solution and the relative performance of Mehlhorn's algorithm (see Figures~\ref{subfig:b}--\ref{subfig:d}).
For this reason, before averaging over the instances, we normalize the solutions costs so that $0$ corresponds to the optimum and $1$ corresponds to the performance of Mehlhorn's algorithm. Specifically, denoting by $c_{\ALPS(\alpha)}, c_{\OPT}$, and $c_{\MST}$ the cost of the solution output by ALPS (with parameter $\alpha$), the optimum value and the cost output by Mehlhorn's algorithm, respectively, we define the normalized cost of ALPS as
$\frac{c_{\ALPS(\alpha)} - c_{\OPT}}{c_{\ALPS(\alpha)}-c_{\MST}}.$
For predictions of varying quality $p$ and for different values of $\alpha$, this measure then accurately reflects how ALPS interpolates between the optimum ($0$) and Mehlhorn's algorithm~($1$).

\begin{figure*}
    \subfloat[Synthetic predictions. Average approximation ratio over all instances.]{%
        \includegraphics[width=.48\linewidth]{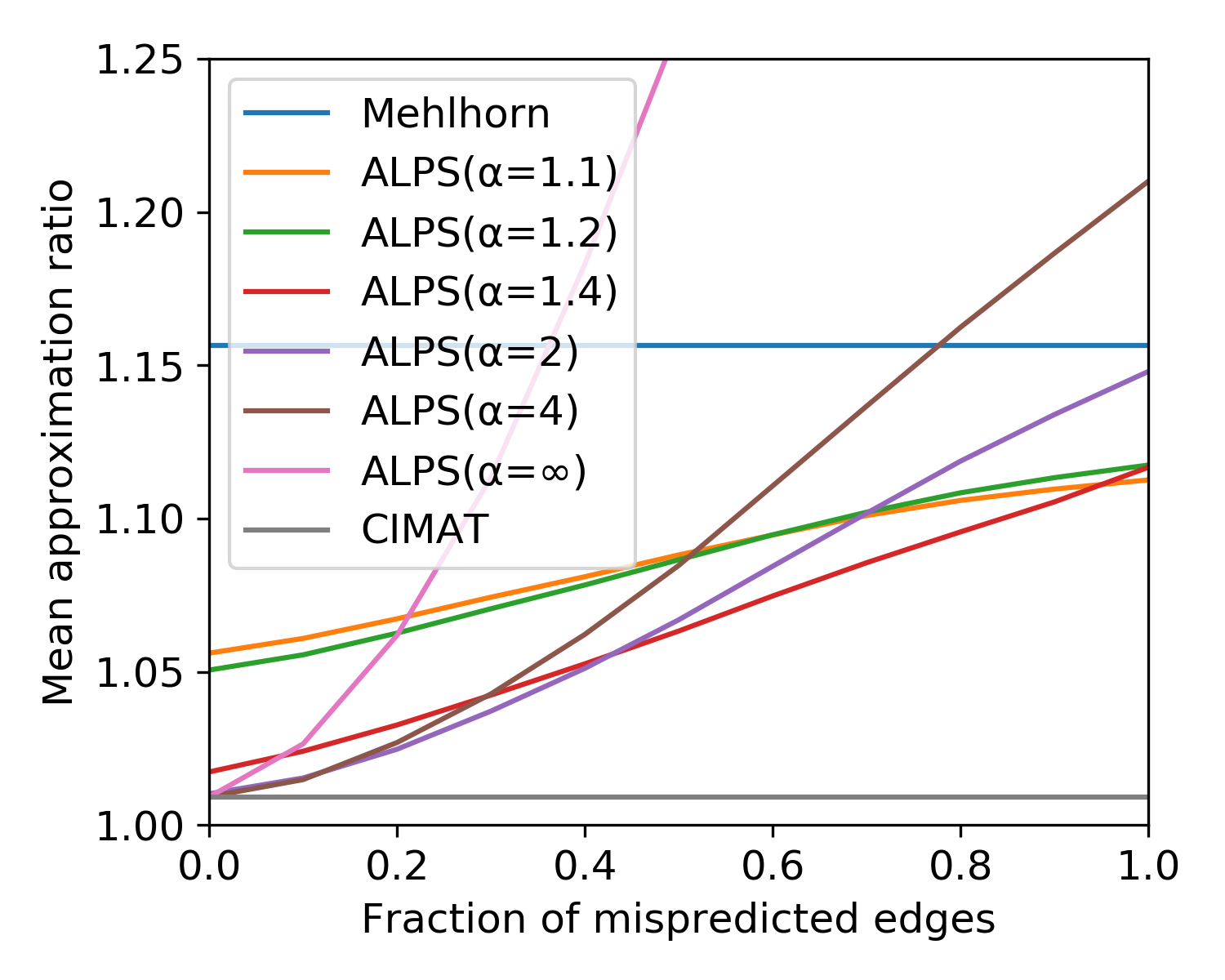}%
        \label{subfig:a}%
    }\hfill
    \subfloat[Learned predictions. ``Fixed core'' distribution based on instance 011.]{%
        \includegraphics[width=.48\linewidth]{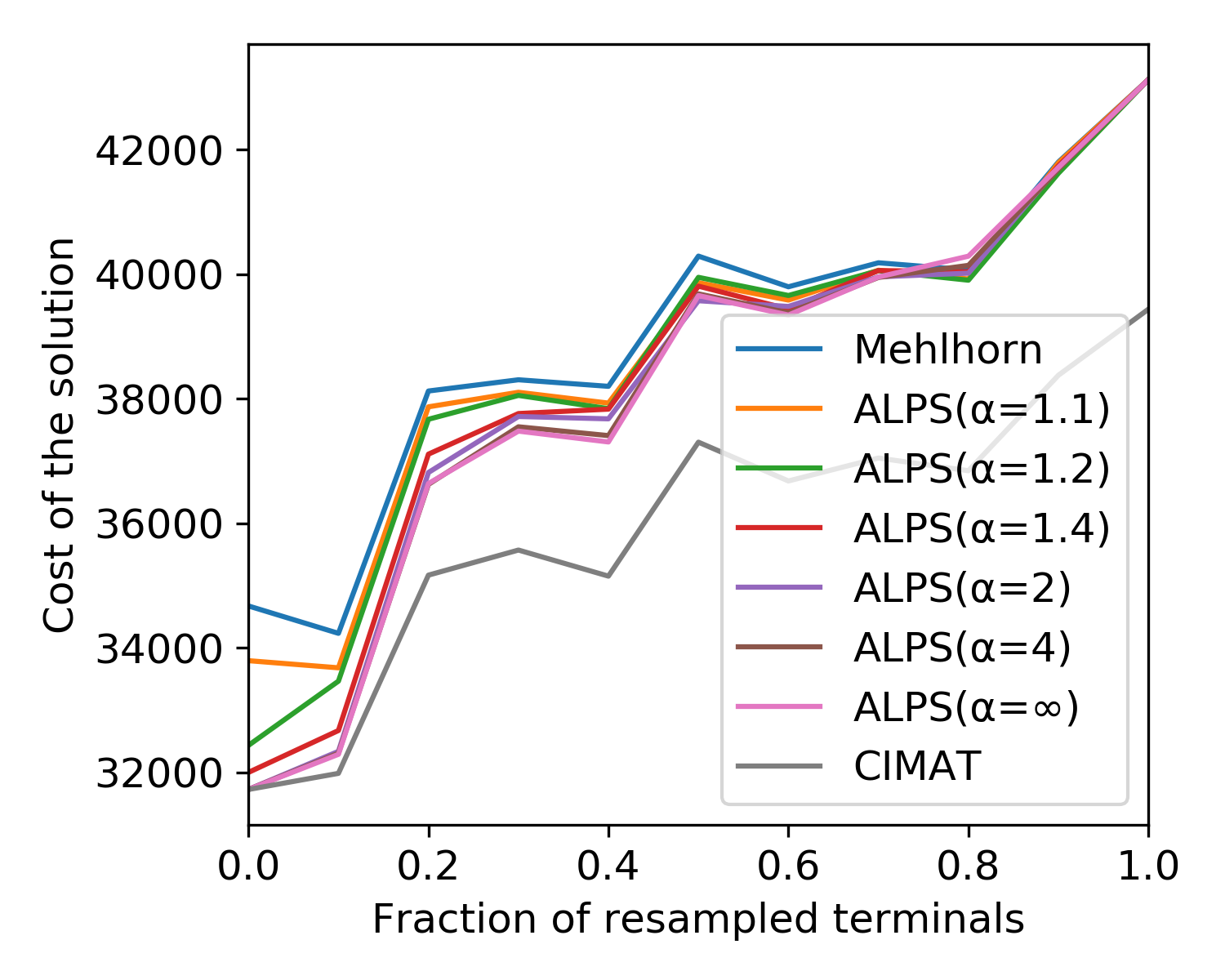}%
        \label{subfig:b}%
    }\\
    \subfloat[Learned predictions. ``No core'' distribution based on instance 082.]{%
        \includegraphics[width=.48\linewidth]{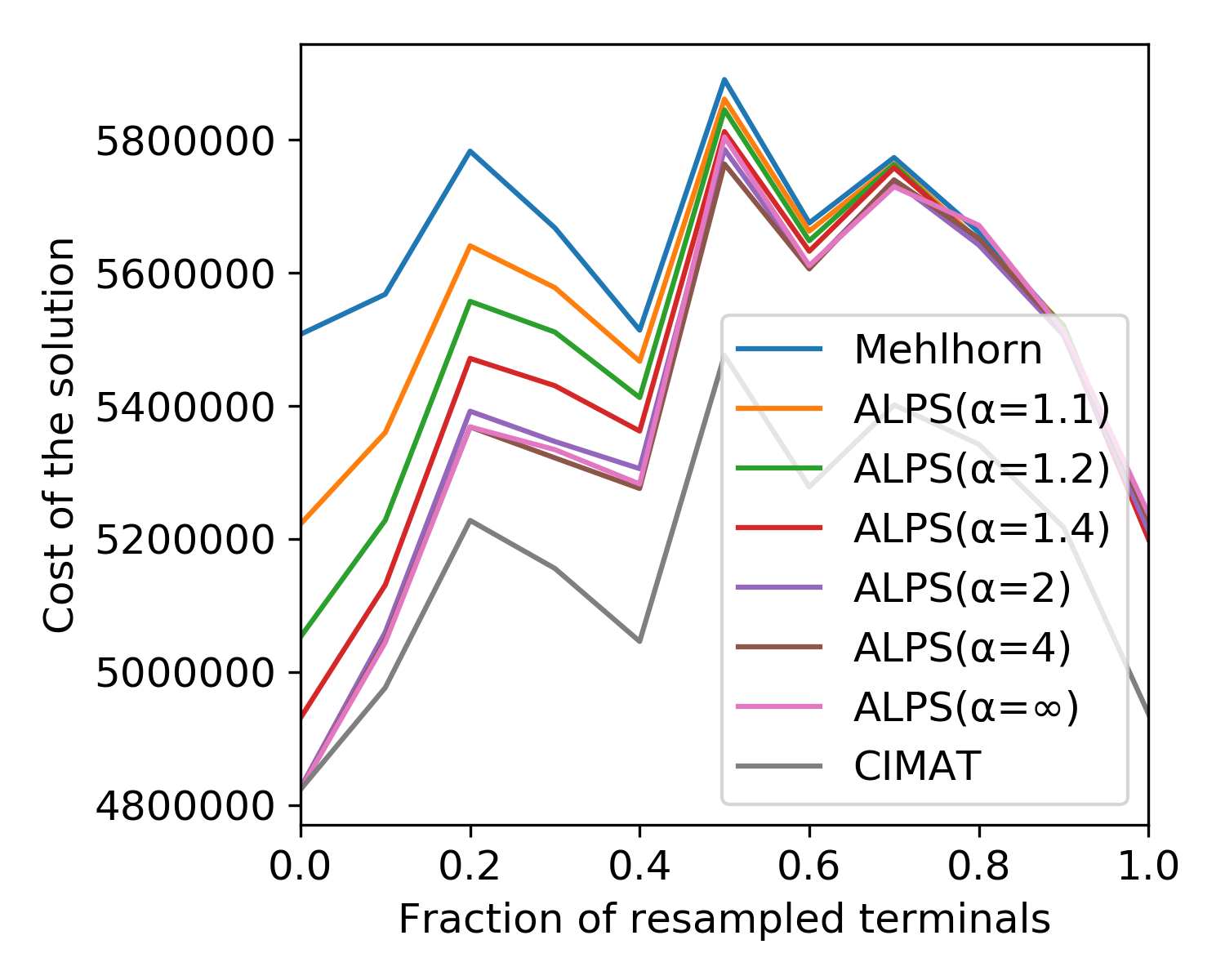}%
        \label{subfig:c}%
    }\hfill
    \subfloat[Learned predictions. ``Fixed core'' distribution based on instance 178.]{%
        \includegraphics[width=.48\linewidth]{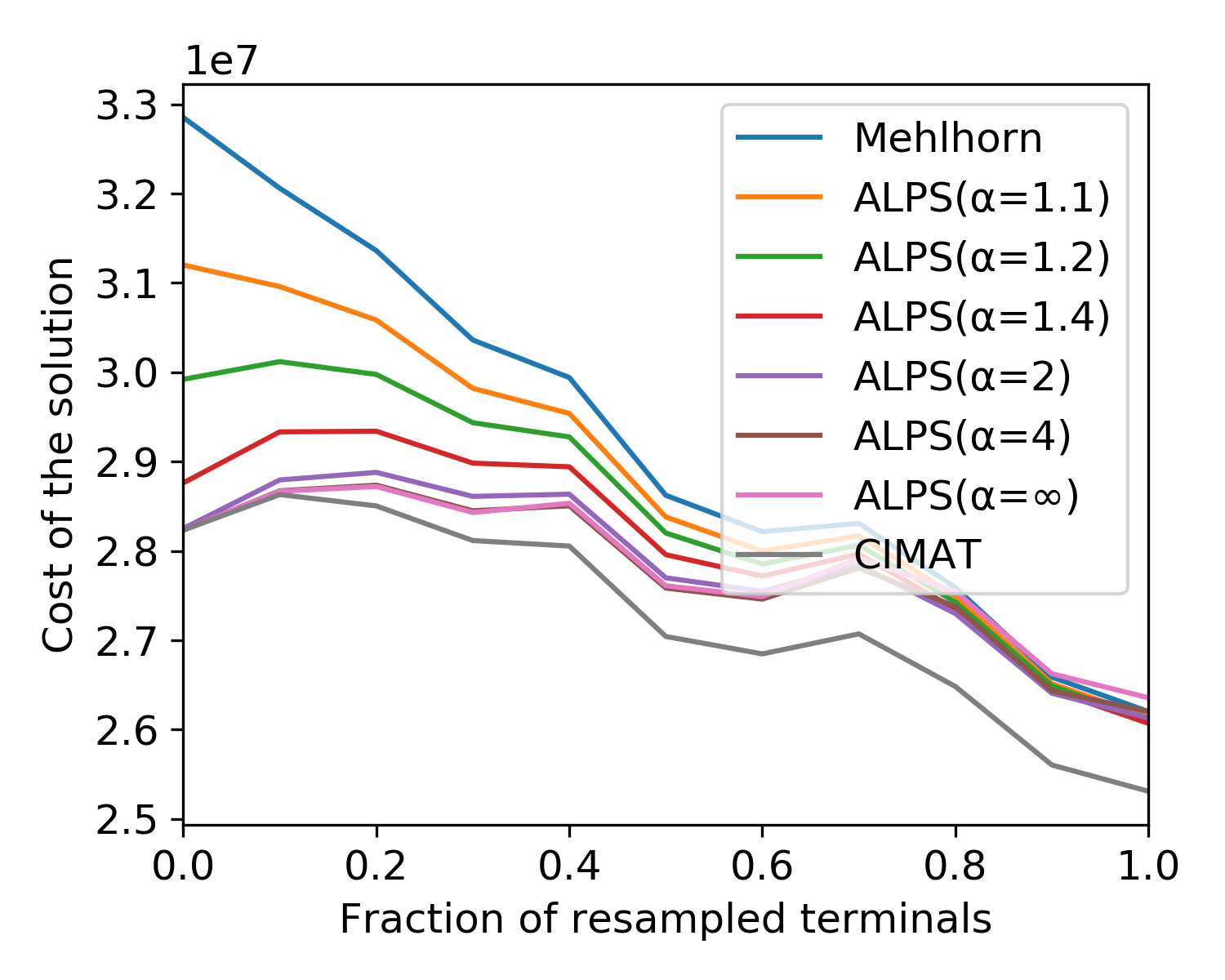}%
        \label{subfig:d}%
    }\\
    \subfloat[Learned predictions. Normalized cost averaged over all distributions with fixed core.]{%
        \includegraphics[width=.48\linewidth]{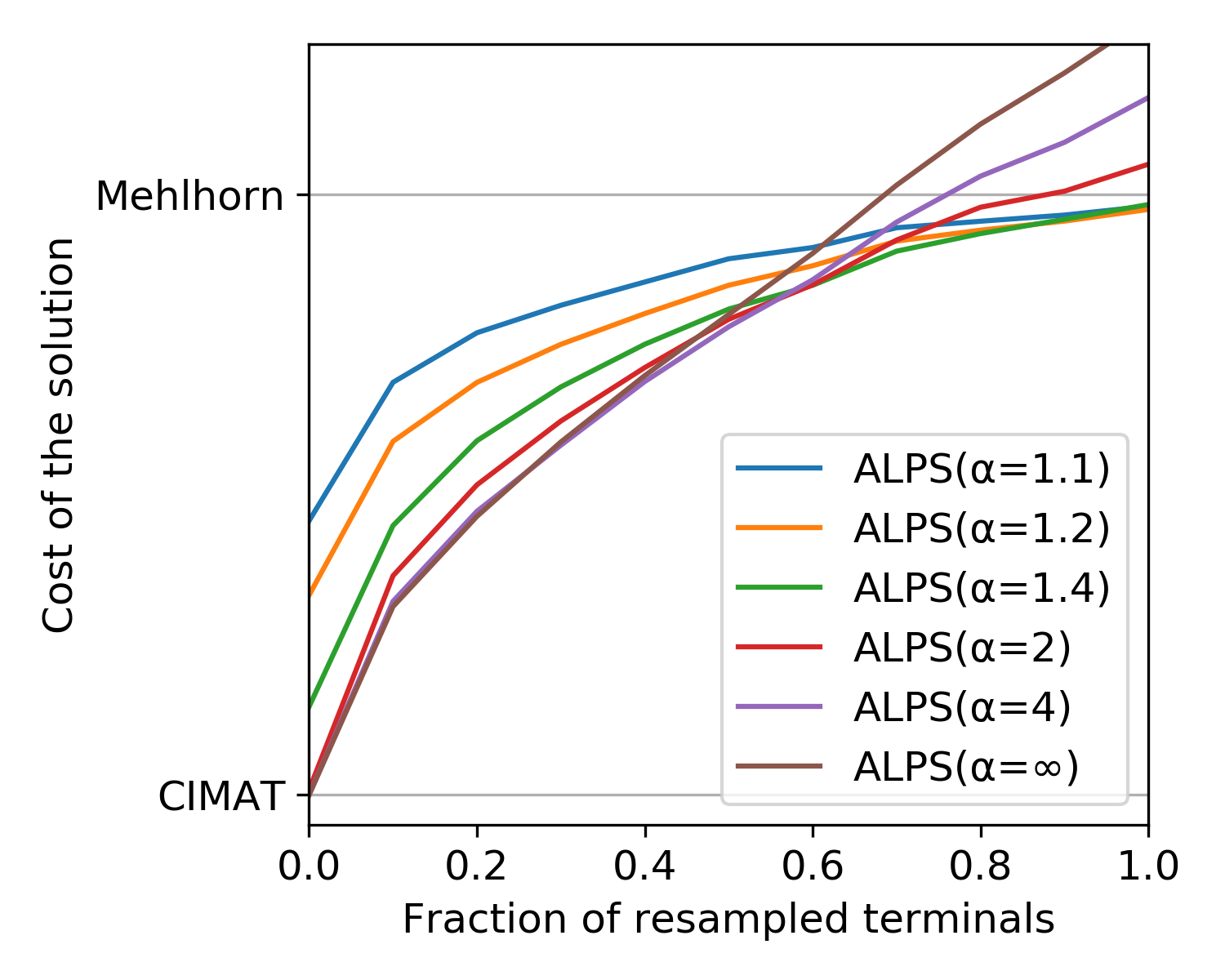}%
        \label{subfig:e}%
    }\hfill
    \subfloat[Learned predictions. Normalized cost averaged over all distributions with no core.]{%
        \includegraphics[width=.48\linewidth]{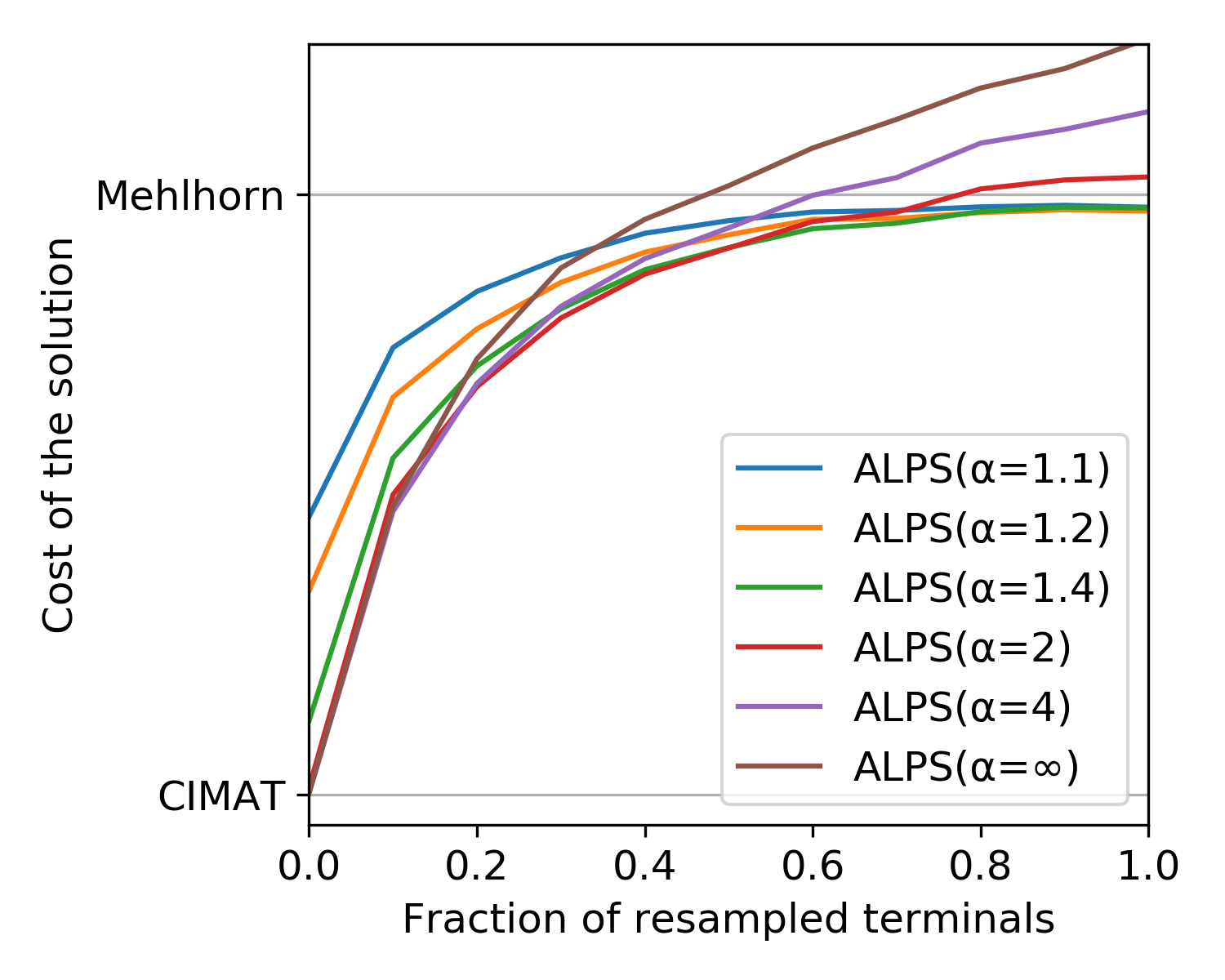}%
        \label{subfig:f}%
    }
    \caption{Experimental evaluation of our refined Steiner Tree algorithm from Section~\ref{sec:steiner-tree} (ALPS), with different values of the confidence parameter $\alpha$, compared to an equally fast classic approximation algorithm (Mehlhorn) and a much slower near-optimal solver (CIMAT). The x-axis represents the parameter $p$ controlling the prediction error: accurate predictions are to the left and erroneous predictions are to the right. The y-axis represents the value of the returned solution (the lower the better).}
    \label{fig:fig}
\end{figure*}

\paragraph{Results.} The results of our empirical evaluation are depicted in Figure~\ref{fig:fig}. Figure~\ref{subfig:a} displays the performance of ALPS with synthetic prediction, averaged over all instances. Figures~\ref{subfig:b},~\ref{subfig:c}, and~\ref{subfig:d} depict several typical behaviors on distributions of the Steiner Tree instances with resampled terminals, which we used in our experiment with learned predictions. Those figures illustrate the need to consider normalized costs. Finally, Figures~\ref{subfig:e} and~\ref{subfig:f} display the normalized costs of algorithms averaged over all distributions, with fixed core and with no core, respectively.

Out of the datapoints that we report, $95\%$ of them had standard deviations below $0.1$ of their values, and the maximum standard deviation among them was $0.25$ of the corresponding value. It would be interesting to run the experiments with a higher number of iterations in order to bring the standard deviations down, but our computational resources do not permit that (the described experiments already required $\approx 700$ CPU hours).

\paragraph{Conclusions.}
In all our experiments, ALPS (especially with the confidence parameter $\alpha \geqslant 1.4$) equipped with good predictions (low values of the parameter $p$, left side of the plots) outputs solutions almost as good as the (near-)optimal solutions output by CIMAT, which is orders of magnitude slower. On the other hand, the classic approximation algorithm of Mehlhorn, which has the same running time as ALPS, outputs solutions with a noticeably higher total cost.

Unsurprisingly, with the increasing prediction error (high values of the parameter $p$, right side of the plots) performance of ALPS slowly degrades. What is notable though is that already for moderate values of the confidence parameter ($\alpha \leqslant 2$) ALPS seems robust, i.e., it is never significantly worse than Mehlhorn, even with bad predictions. This is despite the fact that our ALPS implementation does not involve a separate robustification step (as in Corollary~\ref{cor:robust}).

As the theory predicts, in the case of accurate predictions it is better to choose a large $\alpha$ while for bad predictions a smaller $\alpha$ is better, and, as we explain in Section~\ref{sec:steiner-tree-alpha}, one may want to run ALPS for a geometric progression of $\alpha$'s and pick the best solution. However, in our experiments, it seems that choosing just a single $\alpha=1.4$ or $\alpha=2$ is a good choice almost universally.

Finally, it is worth to note that, somewhat surprisingly, synthetic predictions with $p=1.0$, which are almost completely random, still allow (for small enough $\alpha$) to improve over Mehlhorn's algorithm. This surprising behavior can be explained by the fact that instances that are relatively hardest for the Mehlhorn's algorithm tend to be very symmetric and to have multiple (close-to-)optimal solutions, and it seems that randomly decreasing a small fraction of edge weights by a small factor $\alpha$ breaks the symmetry and guides the algorithm towards a good solution.

\section{Discussion}
We initiated the study of algorithms with predictions with a focus on improving over the approximation guarantees of classic algorithms without increasing the running time. This paper focused on the wide and important class of selection problems, but it would be interesting to investigate whether similar results can be obtained for central combinatorial optimization problems that do not belong to this class, e.g., clustering and scheduling problems or problems with non-linear (e.g., submodular) objectives. 

We demonstrated, with the example of the Steiner Tree problem, that refined algorithms with improved guarantees are possible for specific problems. A second, implicit, advantage of our refined Steiner Tree algorithm is that its actual performance could be bounded in terms of a quantity directly related to the optimal cost, thus avoiding the additional robustification step. An interesting direction for further research would be to identify other problems satisfying these two properties.

\bibliography{submission}
\bibliographystyle{plainnat}

\end{document}